\pdfoutput=1
\documentclass[times]{simauth}

\usepackage{cite}
\usepackage[titletoc,title]{appendix}
\renewcommand\thetable{\Roman{table}}
\definecolor{light-gray}{gray}{0.8}

\usepackage[flushleft]{threeparttable}

\usepackage[T1]{fontenc}
\usepackage{fourier}
\usepackage{float}
\usepackage[english]{babel}											
\usepackage[protrusion=true,expansion=true]{microtype}

\usepackage{amsmath,amsfonts,amsthm}
\newcommand{\indep}{\rotatebox[origin=c]{90}{$\models$}}
\usepackage{enumerate}
\usepackage{diagbox}

\newtheorem{lemma}{Lemma}
\usepackage{dsfont}

\usepackage{multirow}
\usepackage{booktabs}
\usepackage{multirow}
\usepackage{siunitx}

\usepackage{caption}
\usepackage[numbers]{natbib}

\usepackage{hyperref}
\makeatletter
\renewcommand\@biblabel[1]{#1.}
\makeatother

\DeclareMathOperator{\Var}{Var}

\usepackage{graphicx}
\usepackage{epstopdf}

\def\volumeyear{2016}

\begin{document}
\address{\affilnum{a} Department of Statistics, University of Michigan, Ann Arbor, MI 48109, U.S.A.\\
	\affilnum{b} School of Information, University of Michigan, Ann Arbor, MI 48109, U.S.A.}
\corraddr{Peng Liao, 439 West Hall, 1085 South University Ave, Ann Arbor, MI 48109, U.S.A.}
\runninghead{SAMPLE SIZE CALCULATIONS FOR MICRO-RANDOMIZED TRIALS IN MHEALTH}	
\title{Sample Size Calculations for Micro-randomized Trials in mHealth}	
\author{Peng Liao, \affil{a}\corrauth \footnotemark[2] Predrag Klasnja, \affil{b} Ambuj Tewari\affil{a} and Susan A. Murphy\affil{a}}
	
\begin{abstract}
The use and development of mobile interventions are experiencing rapid growth. In \lq\lq just-in-time'' mobile interventions, treatments are provided via a mobile device and they are intended to help an individual make healthy decisions \lq\lq in the moment,'' and thus have a proximal, near future impact. Currently the development of mobile interventions is proceeding at a much faster pace than that of associated  data science methods.  A first step toward developing data-based methods is to provide an experimental design for testing the proximal effects of these just-in-time treatments. In this paper, we propose a \lq\lq micro-randomized'' trial design for this purpose. In a micro-randomized trial, treatments are sequentially randomized throughout the conduct of the study, with the result that each participant may be randomized at the 100s or 1000s of occasions at which a treatment  might be provided. Further, we develop a test statistic for assessing the proximal effect of a treatment as well as an associated sample size calculator. We conduct simulation evaluations of the sample size calculator in various settings. Rules of thumb that might be used in designing a micro-randomized trial are discussed. This work is motivated by our collaboration on the HeartSteps mobile application designed to increase physical activity.
\end{abstract}
	
\keywords{Mirco-randomized Trial, Sample Size Calculation, mHealth}
	
\maketitle	

\footnotetext[2]{\textit{E-mail: pengliao@umich.edu. }}	

\section{Introduction}
	The use and development of mobile interventions are experiencing rapid growth. Mobile interventions are used across the health fields and include treatments to improve HIV medication adherence \citep{Lewis2013, KaplanStone2013}, to increase activity \citep{King2013}, supplement counseling/pharmacotherapy in treatment for substance use \citep{Marsch2012,Boyer2012}, reinforce abstinence in addictions \citep{AlssiandPetry2013, Cucciare2012} and to support recovery from alcohol dependence \citep{Gustafson2014, Quanbeck2014}. Mobile interventions for adherence to anti-retroviral therapy and smoking cessation have shown sufficient effectiveness and replicability in trials and have been recommended for inclusion in health services \citep{Free2013}.
	
	However, as  Nilsen \textit{et al.} \cite{{Nilsen2012}} state, \lq\lq In fact, the development of mHealth technologies is currently progressing at a much faster pace than the science to evaluate their validity and efficacy, introducing the risk that ineffective or even potentially harmful or iatrogenic applications will be implemented.\rq\rq Indeed reviews, while reporting preliminary evidence of effectiveness, call for more programmatic, data-based approaches to constructing mobile interventions \citep{Free2013, Muessig2013}.  In particular, these reviews call for research that focuses on data-informed development of these complex multi-component interventions prior to their evaluation in standard randomized controlled trials. But methods for using data to inform the design and evaluation of adaptive mobile interventions have lagged behind the use and deployment of these interventions \citep{MetzandNilsen2014, Nilsen2012, Kumar2013}.
	
	Many mobile interventions are designed to be \lq\lq just-in-time'' interventions, meaning that they intend to provide treatments that help an individual make healthy decisions in the moment, such as engaging in a desirable behavior (e.g., taking a medication on time) or effectively coping with a stressful situation. As such, mobile interventions are often intended to have proximal, near-term effects. A first approach toward developing data-based methods for evaluation of mobile health interventions is to provide an experimental design for testing the proximal effects of the  treatments. This paper proposes a micro-randomized trial design for this purpose.  In a micro-randomized trial, treatments are sequentially randomized throughout the conduct of the study, with the result that each participant may be randomized at the hundreds or thousands of occasions at which a treatment  might be provided.  This repeated randomization of treatments under investigation enables causal modeling of each treatment's time-varying proximal effect as well as modeling of time-varying effect moderation. Thus, the micro-randomized trial can be seen as a  first experimental step in the development of effective mobile interventions that are composed of sequences of treatments. We propose to size the trial to detect the  proximal main effect of the treatments. This is akin to the use of factorial designs for use in constructing multi-component interventions.  In these factorial designs \citep{Box1978, Chakraborty2009},  a first analysis often involves testing if the main effect of each treatment is equal to 0.
	
	This work is motivated by our collaboration on the HeartSteps mobile application for increasing physical activity, which we will use to illustrate our discussion. One of the treatments in HeartSteps is suggestions for physical activity which are tailored to the person's current context. HeartSteps can deliver these suggestions at any of the five time intervals during the day, which correspond roughly to morning commute, mid-day, mid-afternoon, evening commute, and post-dinner times. When a suggestion is delivered, the user's phone plays a notification sound, vibrates and lights up, and the suggestion is displayed on the lock screen of the phone. These suggestions encourage activity in the current context and are intended to have an effect (getting a person to walk) within the next hour.
	
	In the following section, we introduce the micro-randomized trial design. In section 3 we precisely define the proximal main effect of a treatment,  using the language of potential outcomes.  We develop the test statistic for assessing the proximal effect of a treatment as well as an associated sample size calculator in section 4 and 5.  Next we provide simulation evaluation of the sample size calculator. We end, in Section 7, with a discussion.
	
\section{Micro-Randomized Trial}
	
	In general an individual's longitudinal data, recorded via mobile devices that sense and provide treatments, can be written as
	\begin{align*}
	\{S_0, S_1,  A_1,  S_2,   A_2,  \dots, S_t, A_t,\ldots,S_T, A_T, S_{T+1}\}
	\end{align*}
	where,
	$t$ indexes  decision times,
	$S_0$ is a vector of baseline information (gender, ethnicity, etc.) and
	$S_t(t\geq 1)$ is  information collected between time $t-1$ and $t$ (e.g., summary measures of recent activity levels, engagement, and burden; day of week; weather; busyness indicated by smartphone calendar, etc.). The treatment  at time $t$ is denoted by  $A_t$; throughout this paper we consider binary options for the treatments (e.g., the treatment is on or off). The proximal response, denoted by $Y_{t+1}$, is a known function of $\{S_t, A_t, S_{t+1}\}$. Here we assume that the longitudinal data are independent and identically distributed across $N$ individuals. Note that this assumption would be violated, if for example, some of the treatments are used to enhance social support between individuals in the study.
	
	In HeartSteps, data ($S_t$) is collected both passively via sensors and via participant self-report.  Each participant is provided a \lq\lq Jawbone\rq\rq~band, worn at the wrist, which collects daily step count and the amount of sleep the user had the previous night. Furthermore sensors on the phone are used to collect a variety of information at each of the 5 time points during the day, including  the time-stamp, location, busyness of planned activities on the phone calendar and  other activity on the phone.  Each evening, self-report data is collected including utility and burden ratings.  The proximal response, $Y_{t+1}$, for activity suggestions is the step count in the hour following time $t$.
	
	A decision time is a point in time at which---based on participant's current state, past behavior, or current context---treatment may need to be delivered. Decision times vary by the nature of the intervention component. In HeartSteps, the decision times for activity suggestions are 5 times per day over the 42 day study duration. For an alcohol-recovery application that provides an intervention when an individual goes within 10 feet of a high risk location (e.g., a liquor store), decision points might be every 1 minute, the frequency at which the application would get the person's current location and assess whether she is close to a high-risk location.  In a long-term study of an intervention for multiple health behaviors, the decision points might be weekly or monthly at which times, decisions are made regarding  whether to change the focus from one behavior (e.g., physical activity) to another (e.g., diet). Finally, in many studies there is an option for an individual to press a "panic\rq\rq button, indicating the need for help; for such interventions, decision times correspond to times at which the panic button is pressed.
	
	A  micro-randomized trial is a trial in which at each decision time $t$, participants are randomized to a treatment option, denoted by $A_t$.   Treatment options may correspond to whether or not a treatment is provided at a decision time; for example in HeartSteps, whether or not the individual is provided a lock-screen activity suggestion. Or treatment options may be alternative types of treatment that can be provided at the same decision time; for example, a daily step goal treatment might have two options, a fixed 10,000-steps-a-day goal or an adaptive goal based on the user's activity level on the previous day. Considerations of treatment burden often imply that the randomization will not be uniform.   For example in HeartSteps, the randomization probability is 0.4,  so that, if an individual is always available, on average   $2$ lock-screen activity messages are delivered  per day.

	In designing, that is, determining the sample size for, a micro-randomized trial we focus on the reduced longitudinal data
	\begin{align*}
	\{S_0,  I_1, A_1, Y_2,  I_2,  A_2, Y_3, \dots, I_t, A_t, Y_{t+1},\ldots, I_T, A_T, Y_{T+1}\}.
	\end{align*}
	The variable, $I_t$ is an  \lq\lq availability\rq\rq indicator. The availability indicator is coded as $I_t=1$ if the individual is available for treatment and $I_t=0$ otherwise.  At some decision times feasibility, ethics or burden considerations mean that the individual is unavailable for treatment and thus $A_t$ should not be delivered.  Consider again HeartSteps: if sensors indicate that the individual is likely driving a car or the individual is currently walking, then the lock-screen activity message  should not be sent.  Other examples of when individuals are unavailable for treatment include: in the alcohol recovery setting, an \lq\lq warning\rq\rq treatment  would only be potentially provided when sensors indicate that the individual is within 10 feet of a high risk location or  a treatment might only be provided if the individual reports a high level of craving.  If the application has a panic button, then only in an $x$ second interval in which the panic button is pressed is it appropriate to  provide \lq\lq panic button\rq\rq treatments.   Individuals may be unavailable for treatment by choice.   For example, the HeartSteps application permits the individual to turn off the lock-screen activity messages; this option is considered critical to maintaining participant buy-in and engagement with HeartSteps.  After viewing the lock-screen activity message, the individual has the option of turning off the lock-screen messages for  4 , 8 or 12 hours.  After the specified time interval, the delivery of lock-screen messages automatically turns on again. To summarize, the availability indicator at time $t$ is the indicator for the subpopulation at time $t$ among which we are interested in assessing the proximal main effect of the treatment;   {\it  we are uninterested in assessing the proximal main effect of a treatment among individuals for whom it is unethical to provide treatment or for whom it makes no scientific sense to provide treatment or among those who refuse to be provided a treatment. }

\section{Proximal Main Effect of a Treatment}
	
	As discussed above, treatments in mobile health interventions are often designed so as to have a proximal effect (e.g., increase activity in near future, help an individual manage current cravings for drugs or food, take medications on schedule, etc.).   As a result, a first question in developing a mobile health intervention is whether the treatments have a proximal effect.  Here we develop sample size formulae that guarantee a stated  power to detect the  proximal effect of a treatment.  In particular we aim to test if the  proximal main effect is zero.
	
	To define the proximal main effect of a treatment, we  use potential outcomes \citep{Rubin1978,Robins1986, Robins1987}.  Our use of potential outcome notation is slightly more complicated than usual because treatment can only be provided when an individual is available.  As a result, we index the potential outcomes by decision rules that incorporate availability.   In particular define $d(a,i)$ for $a\in\{0,1\}, \ i\in \{0,1\}$ by $d(a,0)=$\lq\lq unavailable-do nothing\rq\rq and $d(a,1)=a$.  Then for each ${a}_1 \in \mathcal{A}_1 = \{0, 1\}$, define $D_1(a_1)=d(a_1,I_1)$.   Then  we denote the potential proximal responses {following decision time $1$} by $\{Y_2^{D_1(1)},\ Y_2^{D_1(0)}\}$     and denote the potential availability indicators at decision time $2$ by $\{I_2^{D_1(1)},\ I_2^{D_1(0)}\}$.
	Next for each $\bar{a}_2=(a_1, a_2)$ with  $a_1, a_2 \in \{0, 1\}$, define $D_2(\bar a_2)=d(a_2,I_2^{D_{1}(a_{1})})$.  Define $\overline{D_2(\bar a_2)}=(D_1(a_1), D_2(\bar a_2))$.   A potential  proximal response following decision time $2$ and corresponding to $\bar a_2$ is  $Y_3^{\overline{D_{2}(\bar a_{2})}}$ and a potential  availability indicator at decision time $3$ is  $I_3^{\overline{D_{2}(\bar a_{2})}}$.
	Similarly, for each $\bar{a}_t=(a_1, \dots, a_t) \in \mathcal{A}_t = \{(a_1, \dots, a_t) \big| a_i \in \{0, 1\}, i = 1, \dots, t \}$, define $D_t(\bar a_t)=d(a_t,I_t^{\overline{D_{t-1}(\bar a_{t-1})}})$ and $\overline{D_t(\bar a_t)}=(D_1(a_1), \ldots, D_t(\bar a_t))$.   For each $\bar{a}_t=(a_1, \dots, a_t) \in \mathcal{A}_t$, the potential proximal response is  $Y_t^{\overline{D_{t-1}(\bar a_{t-1})}}$ (following decision time $t-1$) and potential  availability indicator is  $I_t^{\overline{D_{t-1}(\bar a_{t-1})}}$ at decision time $t$.

	We define the proximal main effect of a treatment at time $t$ among available individuals  by:
	\begin{align*}
	\beta(t) =E\left( Y_{t+1}^{\overline{D_{t}(\bar A_{t-1},1)}} -Y_{t+1}^{\overline{D_{t}(\bar A_{t-1},0)}}\bigg| I_t^{\overline{D_{t-1}(\bar A_{t-1})}} = 1\right)
	\end{align*}
	where the expectation is taken with respect to the distribution of the potential outcomes and randomization in $\bar A_{t-1}$.  This proximal effect is conditional in that the effect of treatment at time $t$ is defined for only  individuals available for treatment  at time $t$, that is,  {$I_t^{\overline{D_{t-1}(\bar A_{t-1})}} = 1$}. This proximal effect is a main effect in that the effect is marginal over any effects of $\bar A_{t-1}$.  The former conditional aspect of the definition is related to the concept of  viable or feasible dynamic treatment regimes \citep{Wang2012, Robins2004} in which one assesses only the causal effect of  treatments that can actually be provided.

	Consider the proximal main effect, $\beta(t)$, as $t$ varies across time.  $\beta(t)$ may vary across time for a variety of reasons. To see this consider the case of HeartSteps.   Here  $\beta(t)$ might initially increase with increasing $t$ as participants learn and practice the activities  suggested on the lock-screen.  For larger $t$ one might expect to see decreasing or flat $\beta(t)$ due to habituation (participants begin to, at least partially, ignore the messages).  This time variation in $\beta(t)$ can be attributed to both the immediate effect of a lock-screen activity  message as well as interactions between the past lock-screen activity messages and the present activity message; the time variation occurs at least partially  due to the marginal character of $\beta(t)$.    Alternately the conditional definition of $\beta(t)$ means that the effect is only defined among the population of individuals who are available at decision time $t$. Changes in this population may cause changes in $\beta(t)$ across time.   Again consider HeartSteps.  At earlier time points, participants may be highly engaged, yet have not developed habits that in various ways increase their activity, thus most participants will be available.   However as time progresses, some participants may develop sufficiently positive activity habits or anticipate activity suggestions, thus at later decision times these participants may be already active and thus unavailable to receive a suggestion.   Other participants  may become increasing disengaged and repeatedly turn off the lock-screen activity messages; these participants are also unavailable.   Thus as time progresses, $\beta(t)$ may vary due to the subpopulation of participants among whom it is appropriate to assess the effect of the lock-screen activity messages.

	Our main objective in determining the sample size will be to assure sufficient power to detect alternatives to the null hypothesis of no proximal main effect, H$_0:\beta(t)=0,\ t=1,\ldots T$ for a trial with $T$ decision points (if $\beta(t)$ is nonzero then for the population available at decision time $t$, there is a proximal effect).  The proposed test  will be focused on detecting smooth, i.e., continuous in $t$, alternatives to this null hypothesis.
	
	To express $\beta(t)$ in terms of the observed data distribution, we assume consistency  \citep{Robins1986, Robins1987}.  This assumption is that for each $t$, the observed $Y_{t}$ and observed $I_t$ equal the corresponding potential outcomes, $Y_{t}^{\overline{D_{t-1}(\bar a_{t-1})}}$, $I_t^{\overline{D_{t-1}(\bar a_{t-1})}}$ whenever $\bar A_{t-1}=\bar a_{t-1}$.
	This assumption may be violated if some of the treatments promote social linkages between participants, for example, to enhance social/emotional support or to compete in mobile games. In these cases it would be more appropriate to additionally index each individual's potential outcomes by other participants' treatments.
	The micro-randomization plus the consistency assumption implies that the proximal main effect of treatment at time $t$ among available individuals, $\beta({t})$ can be written as,
	\begin{align*}
	\beta(t) &=E\big[ Y_{t+1}^{\overline{D_{t}(\bar A_{t-1},1)}} \big| I_t^{\overline{D_{t-1}(\bar A_{t-1})}} = 1\big] -E\big[ Y_{t+1}^{\overline{D_{t}(\bar A_{t-1},0)}} \big| I_t^{\overline{D_{t-1}(\bar A_{t-1})}} = 1\big]\\
	& = E\big[ Y_{t+1}^{\overline{D_{t}(\bar A_{t-1},1)}} \big| I_t^{\overline{D_{t-1}(\bar A_{t-1})}} = 1, A_t=1\big] -E\big[ Y_{t+1}^{\overline{D_{t}(\bar A_{t-1},0)}} \big| I_t^{\overline{D_{t-1}(\bar A_{t-1})}} = 1, A_t=0\big]\\
	& = E\big[ Y_{t+1}^{\overline{D_{t}(\bar A_{t})}} \big| I_t^{\overline{D_{t-1}(\bar A_{t-1})}} = 1, A_t=1\big] -E\big[ Y_{t+1}^{\overline{D_{t}(\bar A_{t})}} \big| I_t^{\overline{D_{t-1}(\bar A_{t-1})}} = 1, A_t=0\big]\\
	& = E[ Y_{t+1}| I_t = 1, A_t=1 ] -E[ Y_{t+1}| I_t = 1, A_t=0 ]
	\end{align*}
	where the second equality follows from the randomization of the $A_t$'s and the last equality follows from the consistency assumption.

\section{Test Statistic}
	Our sample size formula is  based on a test statistic for use in testing  H$_0:\beta(t)=0,\ t=1,\ldots T$ against a scientifically plausible alternative. This alternative should be formed based on conversations with domain experts.    Here we construct a test statistic to detect alternatives that are, at least approximately, linear in a vector parameter, $\beta$, that is, alternatives of the form $Z_t'\beta$, where the $p\times 1$ vector, $Z_t$, is a function of $t$ and covariates that are unaffected by treatment such as time of day or day of week.  In the case of HeartSteps, a plausible alternative is quadratic:
	\begin{eqnarray}
	\label{beta}
	Z_t' \beta = \Big(1, \Big\lfloor \frac{t-1}{5}\Big\rfloor,(\Big\lfloor \frac{t-1}{5}\Big\rfloor)^2\Big)\beta
	\end{eqnarray}
	where $\beta = (\beta_1,\beta _2 ,\beta _3 )'$ ($p=3$).  Recall that in HeartSteps there are 5 decision times per day; $\lfloor \frac{t-1}{5}\rfloor$ translates decision times $t$ to days.  This rather simplistic parametrization marginalizes across the day and treats the weekends and weekdays similarly.
	
	We propose to use the alternate, {H$_1: \beta(t)=Z_t'\beta$, $t = 1,\ldots,T$} to construct the test statistic.  We base the test statistic on the estimator of $\beta$ in a least squares fit of a working model. A simple working model based on the alternative is:
	\begin{align}
	\label{lsfit}
	E[Y_{t+1}|I_t=1, A_t]= B_t'\alpha  + (A_t -\rho_t)Z_t'\beta
	\end{align}
	over all $t\in \{1,\dots, T\}$, where $\rho_t$ is the known randomization probability ($P[A_t=1]=\rho_t$) and the $q\times 1$ vector $B_t$ is a function of $t$ and covariates that are unaffected by treatment such as time of day or day of week.  Note that $A_t$ is centered by subtracting off the randomization probability; thus the working model for $\alpha(t) = E[Y_{t+1}|I_t=1]$ is $
	B_t'\alpha$. The estimators $\ \hat\alpha,\ \hat\beta$ minimize the least squares error:
	\begin{align}
	\mathbb{P}_N\left\{\sum_{t=1}^T I_t\left(Y_{t+1}-
	B_t'\alpha -
	(A_t -\rho_t)Z_t'\beta\right)^2\right\}
	\label{ls}
	\end{align}
	where $\mathbb{P}_N \big\{ f(X) \big\}$ is defined as the average of $f(X)$ over the sample.
	
	Note that from a technical perspective,  minimizing the least squares criterion, (\ref{ls}), is reminiscent of a GEE analysis \citep{GEE} with identity link function and a working correlation matrix equal to the identity.  Thus it is natural to consider a non-identity working correlation matrix as is common in GEE.  This, however, is problematic from a causal inference perspective.   To see this suppose that the true conditional expectation is in fact $E\left[Y_{t+1}|I_t=1, A_t\right]=B_t'\alpha^* +
	(A_t -\rho_t)Z_t'\beta^*$, that is, the causal parameter, $\beta(t)$ is equal to $Z_t'\beta^*$.  Further suppose that the working correlation matrix has off-diagonal elements and that we estimate $\beta^*$ by minimizing  the weighted (by the inverse of the working correlation matrix) least squares criterion.   In this case the resulting estimating equations include sums of terms such as
	$ I_t  \left(Y_{t+1}-B_t'\alpha -(A_t -\rho_t)Z_t'\beta\right)I_s(A_s-\rho_t)Z_s$  for $t>s$.  Unfortunately, both  availability at time $t$, $I_t$, as well as $Y_{t+1}$ may be affected by treatment in the past (in particular, $A_s$), thus absent strong assumptions $E\left[I_t  \left(Y_{t+1}-B_t'\alpha^* -(A_t -\rho_t)Z_t'\beta^*\right)I_s(A_s-\rho_t)\right]$ is unlikely to be 0.   Recall that a minimal condition for consistency of  estimators of $(\alpha^*,\beta^*)$ is that the estimating equations have expectation 0, thus absent further assumptions, the estimators derived from the weighted least squares criterion are likely biased.
	Another possibility is to include a time-varying variance term in the least squares criterion, that is the $t$th entry in (\ref{ls}) might be weighted by $\sigma_t^{-2}$.  This would be useful in the data analysis, however for sample size calculations,  values of these variances are unlikely to be available.  Thus for  simplicity we use the unweighted least squares criterion in (\ref{ls}).
	
	Assume that the matrices $Q=\sum_{t=1}^TE[I_t]\rho_t(1-\rho_t)Z_tZ_t'$ and $\sum_{t=1}^TE[I_t]B_tB_t'$ are invertible. The least squares estimators, $\hat\alpha$, $\hat\beta$ are consistent estimators of
	\begin{equation}
	\tilde\alpha=
	\left(\sum_{t=1}^TE[I_t]B_tB_t'\right)^{-1}\sum_{t=1}^TE[I_t]\alpha(t)B_t
	\label{tilde.alpha}
	\end{equation} and
	\begin{equation}
	\tilde\beta=
	\left(\sum_{t=1}^TE[I_t]\rho_t(1-\rho_t)Z_tZ_t'\right)^{-1}\sum_{t=1}^TE[I_t]\rho_t(1-\rho_t)\beta(t)Z_t
	\label{tilde.beta}
	\end{equation} respectively. Furthermore if $\beta(t)$ is in fact equal to $Z_t'\beta$ for some $\beta$, then $Z_t'\tilde\beta=\beta({t})$. This is the case even if $E[Y_{t+1}|I_t=1]\neq
	B_t'\tilde\alpha$.  In the Appendix (Lemma \ref{lse}), we prove these results and also show that, under moment conditions,  $\sqrt{N}(\hat\beta-\tilde\beta)$ is asymptotically normal with mean $0$ and variance $\Sigma_\beta=Q^{-1}WQ^{-1}$ where,
	\[
	W=  E\left[\Big(\sum_{t=1}^T \tilde{\epsilon}_tI_t(A_t -\rho_t)Z_t\Big) \times \Big(\sum_{t=1}^T\tilde{\epsilon}_tI_t(A_t -\rho_t)Z_t'\Big)\right]
	\]
	and $\tilde{\epsilon}_t=Y_{t+1}-I_tB_t'\tilde\alpha -(A_t -\rho_t)I_tZ_t'\tilde\beta$. To test the null hypothesis H$_0:\beta(t)=0,\ t=1,\dots, T$,
	one can use a test statistic based on the alternative, e.g.,
	\begin{align}
	\label{teststat}
	N\hat\beta'\hat\Sigma_\beta^{-1}\hat\beta
	\end{align}
	where $\hat\Sigma_\beta=\hat Q^{-1}\hat W\hat Q^{-1}$ and $\hat Q$ and $\hat W$ are plug in estimators. Note that this test statistic results from a GEE analysis with identity link function and a working correlation matrix equal to the identity matrix for which sample size formulae have been developed \citep{Tu2004}. We build on this work as follows. As Tu \textit{et al.} \citep{Tu2004} discuss, under the null hypothesis the large sample distribution of this statistic is a chi-squared with $p$ degrees of freedom  distribution. If $N$, the sample size, is small, then, as recommended by Mancl and DeRouen \cite{ManclandDeRouen2001}, we make small adjustments to improve the small sample approximation to the distribution of the test statistic.  In particular, they recommend adjusting $\hat W$ using the \lq\lq hat\rq\rq~matrix; see the formulae for the adjusted $\hat W$ as well as $\hat Q$  in Appendix~\ref{sec:A1}.  Also in small sample settings, investigators commonly  suggest that instead of  using a critical value based on the chi-squared distribution, a critical value based on the $t-$distribution should be used \citep{LiRedden}.  As we are considering a simultaneous test for multiple parameters we form the critical value based on Hotelling's $T-$squared distribution \citep{Hotelling}. Hotelling's $T-$squared distribution is a multiple of the $F$ distribution given by $\frac{d_2}{d_1(d_1+d_2-1)} F_{d_1, d_2}$; here we use $d_1=p$ and $d_2= N-q-p$ (recall $q$ is the  number of parameters in the nuisance parameter vector, $\alpha$); see the appendix for a rationale.
	In the following, the rejection region for the test of H$_0:\beta(t)=0,\ t=1,\ldots T$ based on  (\ref{teststat}) is
	\begin{align*}
		\left\{ N\hat\beta'\hat\Sigma_\beta^{-1}\hat\beta > \frac{N-q-p}{p(N-q-1)} F_{p, N-q-p}^{-1}\left(1-\alpha_0\right)\right\}
	\end{align*}
	where $\alpha_0$ is the desired significance level.
	
	\section{Sample Size Formulae}
	
	As Tu \textit{et.al} \citep{Tu2004} have developed general sample size formulas in the GEE setting, here we focus on considerations specific to the setting of micro-randomized trials. To size the study, we will determine the sample size needed to detect the alternate, $\beta(t)$ with:
	\begin{align*}
	\text{H}_1:\beta(t)/\bar\sigma = d(t), t = 1,\ldots, T
	\end{align*}
	where $\bar{\sigma}^2=(1/T)\sum_{t=1}^T E\left[\Var\left(Y_{t+1}\big|
	I_t=1,A_t\right)\right]$ is the average  variance and $d(t)$ is a standardized treatment effect. When $N$ is large and  H$_1$ holds, $N\hat \beta^{'}\hat\Sigma_{\beta}^{-1}\hat \beta$ is approximately distributed as a noncentral chi-squared $\chi_p^2(c_N)$, where $c_N$, the non-centrality parameter, satisfies $c_N = N(\bar\sigma \tilde{d})^{'}\Sigma_\beta^{-1}(\bar\sigma \tilde{d})$, and $\tilde{d} = \left(\sum_{t=1}^T E[I_t]\rho_t(1-\rho_t)Z_tZ_t'\right)^{-1}\sum_{t=1}^T E[I_t]\rho_t(1-\rho_t)d(t)Z_t$ \citep{Tu2004}.  Note that $\tilde{d}=\tilde\beta/\bar\sigma$.

	\textbf{Working Assumptions}
	To derive the sample size formula, we use the form of the non-centrality parameter of the limiting non-central chi-squared distribution, along with working assumptions.  The working assumptions are used to simplify the form of $\Sigma_{\beta}^{-1}$.   In particular, we make the following working assumptions:
	\begin{enumerate}[(a)]
		
		\item $E(Y_{t+1}|I_t=1) = B_t'\alpha$, for some $\alpha \in \mathbb{R}^{q}$
		\label{assum:a}
		
		\item $\beta(t) = Z_t'\beta$ for some $\beta \in \mathbb{R}^{p}$
		\label{assum:b}		
		\item $\Var(Y_{t+1}|I_t=1,A_t)$  is constant in $t$ and $A_t$
		\label{assum:c}		
		
		\item $E[\tilde\epsilon_t\tilde\epsilon_s | I_t=1, I_s = 1, A_t, A_s]$ is constant in $A_t$, $A_s$.
		\label{assum:d}		
	\end{enumerate}
	where, as before, $\tilde{\epsilon}_t=Y_{t+1}-I_tB_t'\tilde\alpha -(A_t -\rho_t)I_tZ_t'\tilde\beta$. See appendix \ref{sec:A1} (Lemma \ref{simpli}) for proof of variance formulas under these working assumptions. The above working assumptions are somewhat simplistic but as will be seen below the resulting sample size formula is robust to  moderate violations. First,
	under these working assumptions  the alternative hypothesis can be re-written as
	\begin{align}
	\text{H}_1:\beta/\bar\sigma = d,
	\label{alternative2}
	\end{align}
	where $d$ is a $p$ dimensional vector of standardized effects.
	Furthermore, $\Sigma_\beta$ is given by
	\begin{equation*}
	\Sigma_\beta  = \bar\sigma^2\bigg( \sum_{t=1}^{T}E[I_t]\rho_t(1-\rho_t)Z_tZ_t'\bigg)^{-1},
	\end{equation*}
	and thus $c_N$ is given by
	\begin{equation}
c_N = N{d}'\bigg( \sum_{t=1}^{T}E[I_t]\rho_t(1-\rho_t)Z_tZ_t'\bigg) {d}.
	\label{cN}
	\end{equation}
	To improve the small sample approximation, we use the multiple of the  $F$-distribution as discussed above.
	Thus the sample size, $N$, is found by solving
	\begin{align}
	F_{p, N-q-p; c_N}\left(F_{p, N-q-p}^{-1}\left(1-\alpha_0\right)\right)=1-\beta_0
	\label{samplesizeeq}
	\end{align}
	where $F_{p,N-q-p;c_N}$ is  the noncentral $F$ distribution with noncentrality parameter, $c_N$  and $1-\beta_0$ is the desired power.   The inputs to this sample size formula are $\{Z_t\}_{t=1}^T$, a scientifically meaningful  value for $d$  (see  below for an illustration), the time-varying availability pattern, $\{E[I_t]\}_{t=1}^T$, the desired significance level, $\alpha_0$ and power, $1-\beta_0$.

	Now we describe how the information needed in the sample size formula might be obtained when the alternative is quadratic ($p=3$, (\ref{beta})).  In this case we first elicit the initial standardized proximal main effect given by $Z_1'\beta/ \bar\sigma=\beta_1/\bar\sigma$.
	Second we elicit the averaged across time, standardized proximal main effect $\bar d= \frac{1}{T}\sum_{t=1}^{T}Z_t'\beta/\bar\sigma$.  Lastly we elicit the time at which the proximal main effect is maximal, i.e. $\arg\max_t Z_t'\beta$.   These three quantities can then be used to solve for {$d=( d_1, d_2, d_3)'$}. For example, in HeartSteps,  we might want to determine the sample size to ensure 0.80 power when there is no initial treatment effect on the first day, and the maximum proximal main effect comes around day $29$. We specify the expected availability, $E[I_t]$ to be constant in $t$ and $Z_t$ is given by (\ref{beta}). Table \ref{formulaforheartsteps} gives sample sizes for HeartSteps under a variety of average standardized proximal main effects ($\bar d$).

\begin{table}[H]
	\centering
	\begin{threeparttable}

	\caption{Illustrative sample sizes for HeartSteps. The day of maximal treatment effect is 29. The expected availability is constant in $t$. }
	
	\begin{tabular}{|l|rrrr|}

		\diagbox{$\bar{d}$}{$E[I_t]$}  & 0.7 & 0.6 & 0.5 & 0.4 \\
		\hline
		0.10 &  32 &  36 &  42 &  52  \\
		0.09 &  38 &  44 &  51 &  63  \\
		0.08 &  47 &  54 &  64 &  78  \\
		0.07 &  60 &  69 &  81 & 101 \\
		0.06 &  79 &  92 & 109 & 135 \\
		0.05 & 112 & 130 & 155 & 193 \\
		\hline
	\end{tabular}
	\begin{tablenotes}

 \item   $\bar d=(1/T)\sum_{t=1}^{T}Z_t'd$ is the average standardized treatment effect. 
	\end{tablenotes}
	\label{formulaforheartsteps}
	\end{threeparttable}
\end{table}

	In the behavioral sciences a standardized effect size of $0.2$ is considered small \citep{Cohen1988}.   Thus given the very small standardized effect sizes, the sample sizes given in Table~\ref{formulaforheartsteps} seem unbelievably small. Two points are worth making in this regard.   First the use of the alternative parametric hypothesis (\ref{alternative2})  in forming the test statistic, implies that both between-subject as well as within-subject contrasts in proximal responses are used to detect the alternative.  To see this, note that if we focused on only the first time point, $t=1$, and tested $H_0:\beta(1)=0$, then an appropriate test would be a two-sample $t$-test based on the  proximal response $Y_2$, in which case the required sample size would be much larger (akin to the sample size for a two arm randomized-controlled trial in which 40\% of the subjects are randomized to the treatment arm).   This two-sample $t$-test uses only between-subject contrasts in proximal response to test the  hypothesis.   The required sample size would be even larger for a  test of $H_0:\beta(1)=0,\ \beta(2)=0$ in which no relationship between $\beta(1)$ and $\beta(2)$ is assumed.  Conversely the sample size would be smaller if one focused on detecting alternatives to $H_0:\beta(1)=0,\ \beta(2)=0$ of the form $H_1: \beta(1)=\beta(2) \neq 0$.   The use of the alternative, $\beta(1)=\beta(2) \neq 0$, allows one to construct tests that use both between-subject as well as within-subject contrasts in proximal responses.   Our approach is in between these two extremes in that we focus on detecting smooth, in $t$, alternatives to $H_0:\beta(t)=0$ for all $t$.  This permits use of both within- as well as between-subject contrasts in proximal responses.  The assumption of a parsimonious alternative enables the use of smaller sample sizes.   A second point is that, at this time, there is no general understanding of how large the standardized effect size should be for these "in-the-moment" effects of a treatment.  Thus these standardized effects may or may not be considered small in future.

\section{Simulations}
	We consider a variety of simulations with different generative models to evaluate the performance of the sample size formulae. In the simulations presented here, we use the same setup as in HeartSteps; see Appendix~\ref{sec:A2} for simulations in other setups (Table \ref{worktrue2856}).   Specifically, the duration of the study is 42 days and there are 5 decision times within each day ($T = 210$). The randomization probability is 0.4, i.e. $\rho  = \rho_t = P(A_t = 1) = 0.4$.  The sample size formula is given in (\ref{cN}) and (\ref{samplesizeeq}).  All simulations are based on 1,000 simulated data sets.

	Throughout this section the inputs to this sample size formula are $Z_t=\left(1, \lfloor \frac{t-1}{5}\rfloor,\lfloor \frac{t-1}{5}\rfloor^2 \right)'$, the time-varying availability pattern, $\tau_t=E[I_t]$, $d$, $\alpha_0=.05$ and power, $1-\beta_0=.80$.  The value for the vector $d$ is indirectly specified via (a) the time at which the maximal standardized proximal main effect is achieved ($\arg\max_t Z_t'd$), (b) the averaged across time, standardized proximal main effect $\bar d=\frac{1}{T}\sum_{t=1}^{T}Z_t'd$ and (c) no  initial standardized proximal main effect ($Z_1'd=d_1=0$). The test statistic used to evaluate the sample size formula  is given by (\ref{teststat}) in which $B_t$ and $Z_t$ are set  to $\left(1, \lfloor \frac{t-1}{5}\rfloor,\lfloor \frac{t-1}{5}\rfloor^2 \right)'$.

	The simulation results provided below illustrate that the sample size formula and associated test statistic are robust. For convenience we summarize the results here.  When the working assumptions hold, then under a variety of availability patterns, i.e., time-varying values for $\tau_t=E[I_t]$ (see Figure~\ref{ShapeOfTau}) the desired Type I error and power are preserved.   This is also the case when past treatment impacts availability.  Furthermore the sample size formula is robust to deviations from the working assumptions, that is, provides  the desired Type I error and power; this is true for a variety of forms of the true proximal main effect of the treatment (see Figure~\ref{ShapeOfBeta}), a variety of distributions and correlation patterns for the errors,  and  dependence of $Y_{t+1}$ on past treatment.  In all cases the above robustness occurs as long as we  provide an approximately true or conservative value for the standardized effect, $d$ and if we provide an approximately true or conservative (low) value for the availability, $E[I_t]$.

	In our simulations, we note several areas in which the sample size formula is less robust to the working assumption (c); this is when the error variance in $Y_{t+1}$ varies depending on whether treatment $A_t=1$ or $A_t=0$ or with time $t$.  In particular if the ratio of $\Var[Y_{t+1}|I_t = 1, A_t=1]/\Var[Y_{t+1}|I_t = 1, A_t=0]<1$, then the power is reduced.  Also if average variance, $ E\big[\Var[Y_{t+1}|I_t = 1, A_t] \big]$ varies greatly with time $t$, then the power is reduced.  See below for details. Lastly as would be expected for any sample size formula, using values of the standardized effect size, $d$, or availability   that are larger than the truth degrades the power of the procedure.

\subsection{Working Assumptions Underlying Sample Size Formula are True} 
	First, we considered a variety of settings in which the working assumptions (a)-(d) hold and in which the inputs to the sample size formula are correct ($d$ is correct under the alternate hypothesis and the  time-varying availability $E[I_t]$ is correct). Neither the working assumptions nor the inputs to the sample size formula specify the error distribution, thus in the simulation we consider 5  distributions for the errors in the model for $Y_{t+1}$ including independent normal, Student's $t$ and  exponential distributions as well as two autoregressive (AR) processes; all of these error patterns satisfy $\bar\sigma^2=1$ (recall $\bar{\sigma}^2=(1/T)\sum_{t=1}^T E\left[\Var\left(Y_{t+1}\big| I_t=1,A_t\right)\right]$).  Furthermore neither the working assumptions nor the inputs to the sample size formula specify the dependence of the availability indicator, $I_t$ on past treatment.  Thus we consider  settings in which  the availability decreases as the number of recent treatments increases.  For brevity, we provide these standard results in the Appendix~\ref{sec:A2} (Tables \ref{worktrue42T} and \ref{worktrue42P}).  The results are generally quite good, with very few Type I error rates significantly above .05 and power levels significantly below .80.

 \begin{figure}[H]
  	\centering
  	\includegraphics[width = \linewidth]{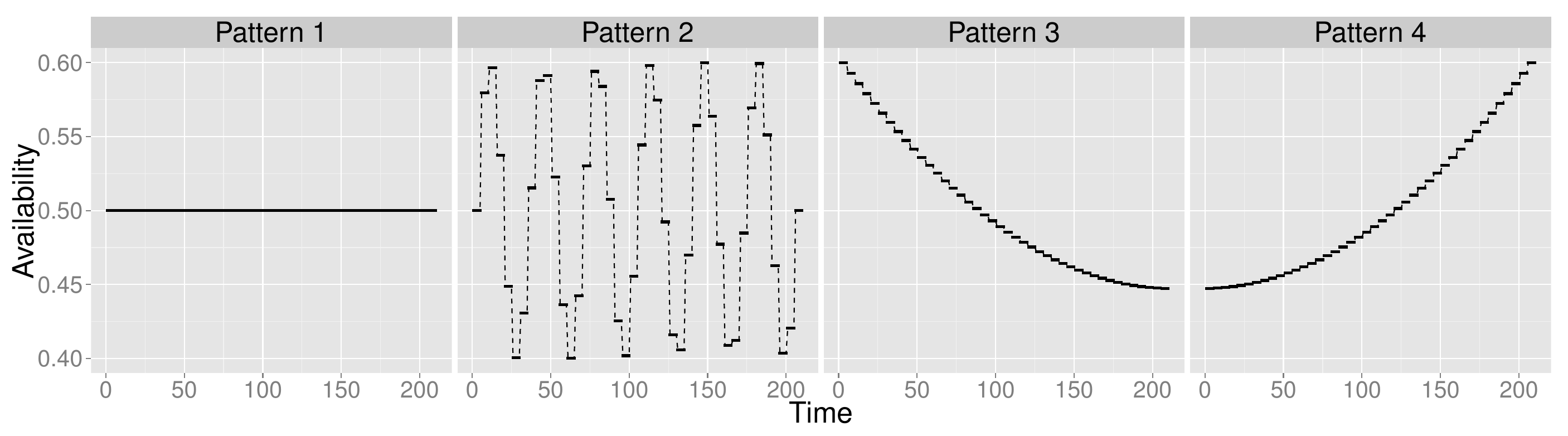}
  	\caption{Availability Patterns.  The x-axis is decision time point and y-axis is the expected availability.  Pattern 2 represents availability varying by day of the week with higher availability on the weekends and lower mid-week. The average availability is 0.5 in all cases. }   	
  	\label{ShapeOfTau}
  \end{figure}
  
\subsection{Working Assumptions Underlying Sample Size Formula are False}
	Second, we considered a variety of settings in which the working assumptions are false but the inputs to the sample size formula are approximately correct as follows.   Throughout  $\bar\sigma^2=1$.

\subsubsection{Working Assumption (a) is Violated.}

	Suppose that the true $E[Y_{t+1}|I_t = 1] \neq B_t\alpha$ for any $\alpha \in \mathbb{R}^q$. In particular,  we consider the scenario in which there is a "weekend" effect on $Y_{t+1}$; see other scenario in Appendix \ref{sec:A2}. The data is generated as follows,
	\begin{align*}
	& I_t \stackrel{Ber}{\sim} \big(\tau_t\big), \ \ A_t \stackrel{Ber}{\sim} \big(\rho \big)\\
	& Y_{t+1} = \alpha(t) +  (A_t - \rho)Z_t'd + \epsilon_t, \text{ if $I_t = 1$}
	\end{align*}
	where the conditional mean $\alpha(t) = B_t'\alpha + W_t \theta$. $W_t$ is a binary variable: $W_t = 1$ if day of the week is time $t$ is a weekend day, and $W_t = 0$ if the day is a weekday. For simplicity, we assume each subject starts on Monday, e.g., for $k = 1, \dots, 6$, $W_{i + 35(k-1)} = 0$, when $i = 1, \dots, 25$, $W_{i + 35(k-1)} = 1$, when $i = 26, \dots, 35$ (recall that we assume in the simulation that there are 5 decision time points per day and the length of the study is 6 week). The values of $\{\alpha_i, i = 1,2, 3\}$ are determined by setting $\alpha(1) = 2.5, \arg \max_t \alpha(t) = T, (1/T)\sum_{t=1}^{T}\alpha(t) - \alpha(1)  = 0.1 $. The error terms $\{\epsilon_t\}_{t=1}^N$ are i.i.d. N$(0, 1)$. The day of maximal proximal effect is  29.  Additionally, different values of the averaged standardized treatment effect and four patterns of availability as shown in Figure \ref{ShapeOfTau} with average 0.5 and are considered.  The type I error rate is not affected, thus is omitted here. The  simulated power is reported in Table \ref{workfalsecase62P}; for more details see Table \ref{workfalsea1} in Appendix \ref{sec:A2}.

\begin{table}[H]

	\centering	
	\begin{threeparttable}
	\caption{Simulated power when working assumption (a) is violated. The patterns of availability are provided in Figure \ref{ShapeOfTau}. }

	\begin{tabular}{|c|c|ccc|}

		&&\multicolumn{3}{c|}{Availability Pattern}\\
		\cline{3-5}	
				{$\theta$}&{$\bar d$} & Pattern 1 & Pattern 2 & Pattern 3 \\
		\hline
		\multirow{2}{*}{$0.5 \bar d$} & 0.10 & 0.80 & 0.79 & 0.81 \\
		& 0.06 & 0.78 & 0.83 & 0.81 \\
		\hline
		\multirow{2}{*}{$ 1\bar d$}& 0.10 & 0.79 & 0.78 & 0.78 \\
		& 0.06 & 0.78 & 0.79 & 0.79 \\
		\hline
		\multirow{2}{*}{$1.5 \bar d$} & 0.10 & 0.78 & 0.81 & 0.78 \\
		& 0.06 & \textbf{0.77} & 0.81 & 0.82 \\
		\hline
		\multirow{2}{*}{$2 \bar d$} & 0.10 & 0.78 & 0.79 & 0.79 \\
		& 0.06 & 0.81 & 0.79 & 0.78 \\
		\hline
	\end{tabular}
	\begin{tablenotes}
		\item   $\theta$ is the coefficient of $W_t$ in $E[Y_{t+1}|I_t =1]$. $\bar d=(1/T)\sum_{t=1}^{T}Z_t'd$ is the average standardized treatment effect.  Bold numbers are significantly (at .05 level) greater lower than 0.80. 
	\end{tablenotes}
	\label{workfalsecase62P}
	\end{threeparttable}
\end{table}

\subsubsection{Working Assumption (b) is Violated.}

	Suppose that the true $\beta(t)\neq Z_t'\beta$ for any $\beta$. Instead the vector of standardized effect, $d$, used in the sample size formula corresponds to the projection of $d(t)$, that is,
	\noindent $d=  \left(\sum_{t=1}^TE[I_t]Z_tZ_t'\right)^{-1}\sum_{t=1}^TE[I_t]Z_td(t)$  (recall $d(t)=\beta(t)/\bar\sigma$ and $\rho_t = \rho$).   The sample size formula is used with the correct availability pattern, $\{E[I_t]\}_{t=1}^T$.
	The data for each simulated subject is generated sequentially as follows.
	For each time $t$,
	\begin{align*}
	& I_t \stackrel{Ber}{\sim} \big(\tau_t\big), \ \ A_t \stackrel{Ber}{\sim} \big(\rho \big)\\
	& Y_{t+1} = \alpha(t) + (A_t - \rho)d(t) + \epsilon_t, \text{ if $I_t = 1$}
	\end{align*}
	for the variety of $d(t)=\beta(t)/\bar\sigma$ and $E[I_t]$ patterns provided in Figure~\ref{ShapeOfBeta} and in Figure~\ref{ShapeOfTau} respectively. The average availability is 0.5. The error terms $\{\epsilon_t\}_{t=1}^T$ are generated as i.i.d. $N(0,1)$.  The conditional mean, $E[Y_{t+1}|I_t=1]=\alpha(t)$ is given by $\alpha(t) = \alpha_1 + \alpha_2 \lfloor \frac{t-1}{5} \rfloor + \alpha_3 \lfloor \frac{t-1}{5} \rfloor ^2$, where $\alpha_1 = 2.5$, $\alpha_2 = 0.727$,$\alpha_3 = - 8.66\times 10^{-4}$ (so that $(1/T)\sum_t\alpha(t) - \alpha(1)= 1$, $\operatorname{arg\,max}_t\alpha(t) = T$).

\begin{figure}[H]

	\centering
	\includegraphics[width = 1\linewidth]{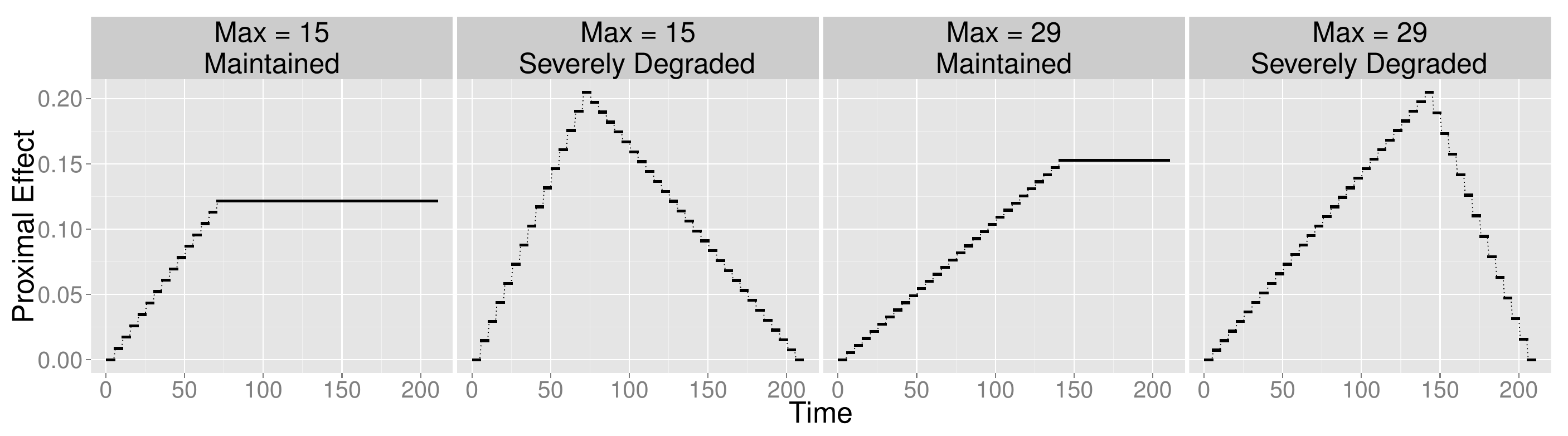} 
		\caption{Standardized Proximal Main Effects of Treatment, $\{d(t)\}_{t=1}^T$: representing maintained and severely degraded time-varying proximal treatment effects. The horizontal axis is the decision time point. The vertical axis is the standardized treatment effect. The "Max" in the titles refer to the day of maximal proximal effect. The average standardized proximal effect is $\bar d= 0.1$ in all plots.
		}
	\label{ShapeOfBeta}
\end{figure}

	The simulated powers are provided in Table \ref{workfalsecase1}. In all cases the power is close to $.80$; this is because all of the proximal main effect patterns in Figure~\ref{ShapeOfBeta} are sufficiently well approximated by a quadratic in time. See Appendix~\ref{sec:A2} for other cases of $d(t)$ and details (Figure \ref{ShapeOfBetafull} and Table \ref{workfalsecase1full}).

\begin{table}[H]
	\centering
	\begin{threeparttable}
	\caption{ Simulated power when working assumption (b) is violated. The shape of the standardized proximal effect and pattern for availability are provided in Figure~\ref{ShapeOfBeta} and \ref{ShapeOfTau} respectively. The sample sizes are given on the right.  }	
	
	\begin{tabular}{|c|c|c|cc|rr|}
		& & & \multicolumn{2}{c|}{Shape of $d(t)$} & & \\
		\cline{4-5}		
		$\bar d$& Availability Pattern & Max  & Maintained & Degraded &	\multicolumn{2}{c|}{Sample Size}  \\
		
		\hline
		\multirow{6}{*}
		{0.10} & \multirow{2}{*}{Pattern 1} & 15 & 0.78 & 0.79 & 43& 39\\
		&  & 29 & 0.80 & 0.79 & 38& 38\\
		\cline{2-7}
		&\multirow{2}{*}{Pattern 2} & 15 & 0.79 & 0.80 & 43 & 39\\
		&  & 29 & 0.78 & 0.79 &  38 &  38 \\
		\cline{2-7}
		&\multirow{2}{*}{Pattern 3} & 15 & 0.81 & \textbf{0.77} &  45 &  41 \\
		&    & 29 & 0.81 & 0.78 &  37 &  39 \\
		\hline
		\multirow{6}{*}
		{0.06 } &\multirow{2}{*}{Pattern 1} & 15 & 0.81 & 0.79 & 111 & 100 \\
		&   & 29 & 0.81 & 0.79  &  96 &  96 \\
		\cline{2-7}	
		& \multirow{2}{*}{Pattern 2}& 15 & 0.79 & 0.81 & 112 & 100 \\
		&  & 29 & 0.79 & 0.80 &  96 &  96 \\
		\cline{2-7}	
		& \multirow{2}{*}{Pattern 3}& 15 & 0.78 & 0.81 & 116 & 106 \\
		&   & 29 & 0.80 & 0.80 &  95 & 101 \\
		\hline
	\end{tabular}
	\begin{tablenotes}

		\item   $\bar d=(1/T)\sum_{t=1}^{T}Z_t'd$ is the average standardized treatment effect. The "Max" in the first row refers to the day of maximal proximal effect. Bold numbers are significantly (at .05 level) lower than .80. 
	\end{tablenotes}
	\label{workfalsecase1}
	\end{threeparttable}
\end{table}

\subsubsection{Working Assumption (c) is Violated.}
	Suppose that $\Var[Y_{t+1}|I_t = 1, A_t]=A_t\sigma_{1t}^2 + (1-A_t)\sigma_{0t}^2$ where $\sigma_{1t}/\sigma_{0t} \neq 1$. The sample size formula is used with the correct pattern for $\{Z_t'd,\ E[I_t]\}_{t=1}^T$. The data for each simulated subject is generated sequentially as follows. For each time $t$,
	\begin{align*}
	& I_t \stackrel{Ber}{\sim} \big(\tau_t\big), \ \ A_t \stackrel{Ber}{\sim} \big(\rho \big)\\
	& Y_{t+1} = \alpha(t) + (A_t - \rho)Z_t'd + \mathds{1}_{\{A_t = 1\}}\sigma_{1t}\epsilon_t + \mathds{1}_{\{A_t = 0\}}\sigma_{0t}\epsilon_t, \text{ if $I_t = 1$}
\end{align*}
	where the average across time standardized proximal main effect, $\bar d=\frac{1}{T}\sum_{t=1}^{T}Z_t'd$ is $0.1$ and day of maximal effect is equal to 22 or 29. The function $\alpha(t)=E[Y_{t+1}|I_t=1]$ is as in the prior simulation. The availability, $\tau_t=0.5$. The error terms $\{\epsilon_t\}$ follow a normal AR(1) process, e.g., $\epsilon_t = \phi\epsilon_{t-1} + v_t$ with the variance of $v_t$ scaled so that $\Var[\epsilon_t] = 1$. Define $ \bar \sigma_t^2= E\big[\Var[Y_{t+1}|I_t = 1, A_t] \big]\left(=\rho\sigma_{1t}^2 + (1-\rho)\sigma_{0t}^2\right) $.  Recall the average variance $\bar\sigma^2$ is given by $(1/T)\sum_{t=1}^T\bar\sigma_t^2$. We consider 3 time-varying trends for $\{\bar \sigma_t\}$ together with different values of $\sigma_{1t}/\sigma_{0t}$; see Figure (\ref{TrendofSigma}). In each trend, $\bar \sigma_t^2$ is scaled such that $\bar \sigma = 1$; thus the standardized proximal main effect  in the generative model is $Z_t'd$. In all cases, the simulated type I error rates are close to $.05$ and thus the table is omitted here (see Appendix~\ref{sec:A2}, Table~\ref{workfalsecase2T}).  The simulated  power is given in Table \ref{workfalsecase2P}.

 \begin{figure}[H]

	\centering
	\includegraphics[width = 0.8\linewidth]{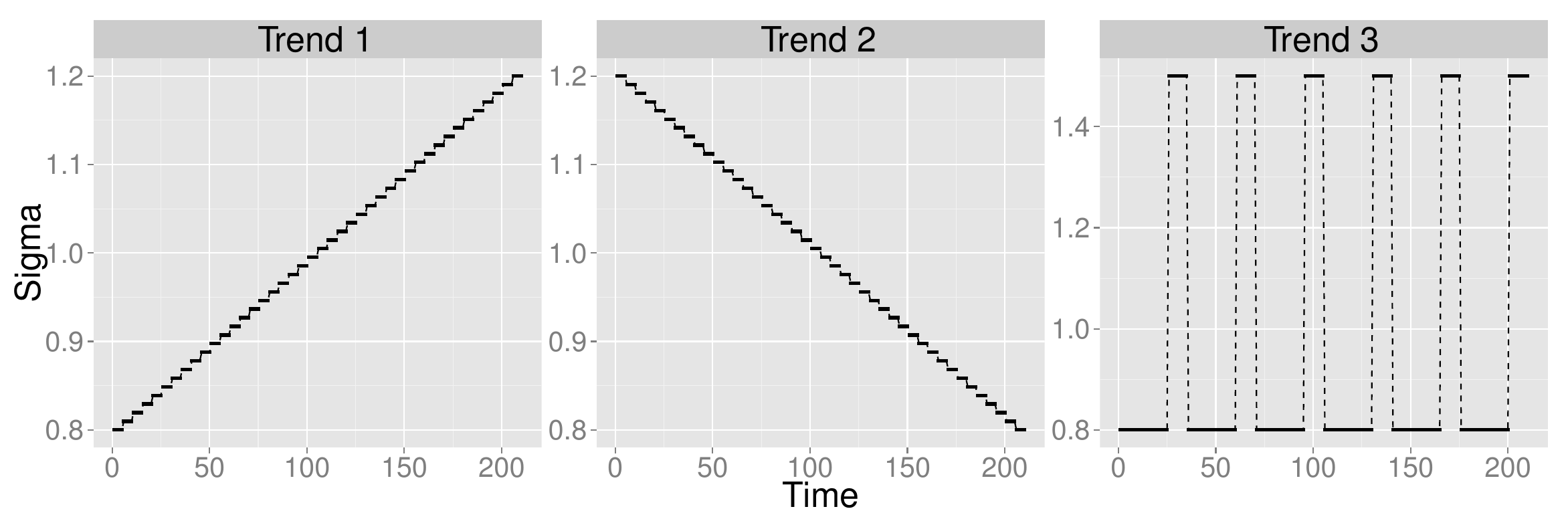}
	\caption{Trend of $\bar \sigma_t$: For all trends, $\bar \sigma_t^2$ is scaled so that $(1/T)\sum_{t=1}^{T}\bar \sigma_t^2 = 1$.  In Trend 3, the  variance, $\bar \sigma_t^2= E\big[Var[Y_{t+1}|I_t = 1, A_t] \big]$ peaks on weekends. In particular, $\bar \sigma_{7k+i} = 0.8$ for $i = 1,\dots,5$ and $\bar \sigma_{7k+i} = 1.5$ for $i = 6, 7$.}		
	\label{TrendofSigma}
\end{figure}

	In the case of $\sigma_{1t} < \sigma_{0t}$, the simulated powers are slightly {larger} than 0.8, while the simulated powers are {smaller} than 0.8 in the case of $\sigma_{1t} > \sigma_{0t}$. The impact of $\bar \sigma_t$ on the power depends on the shape of treatment effect: when $\beta(t)$ attains its maximum, more than halfway through the study, at day 29, a increasing $\{\bar\sigma_t\}$, trend 1,  lowers the power, while a decreasing $\{\bar\sigma_t\}$, trend 2, improves the power. When $\beta(t)$ attains a maximal effect midway through the study, either decreasing or increasing $\{\bar \sigma_t\}$ does not impact power. A large variation in $\bar\sigma_t$, e.g., trend 3, reduces the power in all cases. The differing auto correlations of the errors,  $\epsilon_t$, do not affect power; see a more detailed table in Appendix \ref{sec:A2}, Table~\ref{workfalsecase2T}.

\begin{table}[H]
	\centering
	\begin{threeparttable}
	\caption{Simulated power when working assumption (c) is violated, $\sigma_{1t}\neq \sigma_{0t}$. The trends are provided in Figure~\ref{TrendofSigma}. The availability is 0.5. The average proximal main effect, $\bar d= 0.1$ and the day of maximal effect is 22 or 29, and thus the associated sample sizes are 41 and 42. }	
	\begin{tabular}{|c|c|ccc|ccc|}
		& &
		\multicolumn{3}{c|}{Max = 22 (N = 41)} &
		\multicolumn{3}{c|}{Max = 29 (N = 42)}  \\
		$\phi$ & $\frac{\sigma_{1t}}{\sigma_{0t}}$& trend 1 & trend 2 & trend 3 & trend 1 & trend 2 & trend 3 \\
		\hline
		& 0.8 & 0.83 & 0.84 & 0.80 & 0.81 & 0.89 & 0.79 \\
		-0.6 & 1.0 & 0.79 & 0.80 & \textbf{0.75} & \textbf{0.74} & 0.85 & \textbf{0.70} \\
		& 1.2 & \textbf{0.76} &\textbf{ 0.76} & \textbf{0.71} & \textbf{0.72} & 0.81 & \textbf{0.70} \\
		\hline
		& 0.8 & 0.85 & 0.82 & 0.79 & 0.81 & 0.88 & 0.78 \\
		0 & 1.0 & 0.79 & 0.81 & \textbf{0.74} &\textbf{ 0.77} & 0.86 & \textbf{0.72} \\
		& 1.2 & \textbf{0.77} & \textbf{0.77 }& \textbf{0.71} & \textbf{0.70} & 0.83 & \textbf{0.70} \\
		\hline
		& 0.8 & 0.83 & 0.83 & 0.81 & \textbf{0.77} & 0.87 & \textbf{0.77} \\
		0.6 & 1.0 & \textbf{0.76} & 0.79 & \textbf{0.75} & \textbf{0.73} & 0.85 & \textbf{0.77} \\
		& 1.2 & 0.78 & \textbf{0.77 }& \textbf{0.73} & \textbf{0.72} & 0.82 & \textbf{0.69} \\
		\hline
	\end{tabular}
	\label{workfalsecase2P}

	\begin{tablenotes}
		\item   $\phi$ is the parameter in AR(1) for $\{\epsilon_t\}_{t=1}^T$. \lq\lq Max\rq\rq is the day in which the maximal proximal effect is attained. Bold numbers are significantly (at .05 level) lower than .80.
	\end{tablenotes}
	\end{threeparttable}	
\end{table}

\begin{table}[H]
	\centering
	\begin{threeparttable}
		
	\caption{Simulated power when working assumption (d) is false. The expected availability is 0.5, the average proximal main effect $\bar d=0.1$ and the maximal effect is attained at day 29. The associated sample size is 42. }	

	\begin{tabular}{|c|c|ccc|}
		Parameters in $I_t$ &\diagbox{$\gamma_1$}{$\gamma_2$}&  -0.1 &  -0.2 &  -0.3 \\
		\hline
		&-0.2 & 0.80 & 0.81 & 0.79 \\
		$\eta_1 = -0.1, \eta_2 = -0.1$&-0.5 & 0.79 & 0.81 & 0.80 \\
		&-0.8 & 0.81 & 0.82 & 0.79 \\
		\hline
		&-0.2 & 0.78 & 0.82 & 0.79 \\
		$\eta_1 = -0.2, \eta_2 = -0.1$&-0.5  & 0.81 & \textbf{0.77} & \textbf{0.77} \\
		&-0.8 & 0.81 & 0.79 & 0.78 \\
		\hline
		&-0.2 & 0.78 & 0.78 & 0.80 \\
		$\eta_1 = -0.1, \eta_2 = -0.2$&-0.5 & 0.80 & 0.79 & 0.78 \\
		&-0.8 & 0.78 & 0.79 & 0.80 \\
		\hline
	\end{tabular}
	\begin{tablenotes}
		\item   $\gamma_1$ and $\gamma_2$ are parameters for the cumulative treatments in the model of $Y_{t+1}$. $\eta_1$ and $\eta_2$ are parameters in the model of $I_t$. Bold numbers are significantly (at .05 level) less than .80.
	\end{tablenotes}
	\label{workfalsecase3P}
	\end{threeparttable}	
\end{table}

\subsubsection{ Working Assumption (d) is Violated}.  
	We violate assumption (d) by  making both the availability indicator, $I_t$ and proximal response, $Y_{t+1}$ depend on past treatment and past proximal responses. The sample size formula is used with the correct value of $\{Z_t'd, E[I_t]\}_{t=1}^T$; in particular $d$ is determined by an average proximal main  effect of $\bar d= 0.1$, day of maximal effect equal to 29 ($d_1 =0, d_2 = 9.64 \times 10^{-3} , d_3 = -1.72 \times  10^{-4}$) and with a constant  availability pattern equal to 0.5. The data for each simulated subject is generated as follows. Denote the cumulative treatment over last 24 hours by $C_t = \sum_{j = 1}^{5}A_{t-j}I_{t-j}$.  In each time $t$,
		\begin{align*}
	& I_t \stackrel{Ber}{\sim} \big(\tau_t +  \tau_t\eta_1 (C_t - E[C_t]) + \tau_t\eta_2 \operatorname{Trunc}(\frac{1}{5}\sum_{j = 1}^{5}\epsilon_{t-j})\big),  \ \ A_t \stackrel{Ber}{\sim} \big(\rho \big)\\
	& Y_{t+1} = \begin{cases}
	\alpha(t) + \gamma_1 \left[C_t-E[C_t | I_t = 1]\right] + (A_t - \rho)\big[Z_t'd + Z_t'd\gamma_{2}(C_t-E[C_t|I_t = 1])\big] + \sigma^* \epsilon_t \text{ if $I_t$ = 1}\\
	\alpha_0(t)  + \epsilon_t \text{ if $I_t $ = 0}.
	\end{cases}
	\end{align*}

	where $\{\epsilon_t\}_{t=1}^T$ are i.i.d $N(0,1)$ and $\operatorname{Trunc}(x) := x\mathds{1}_{|x| \leq 1} + \operatorname{sign}(x)\mathds{I}_{|x| >1}$ (the truncation is used to ensure that $\tau_t +  \tau_t\eta_1 (C_t - E[C_t]) + \tau_t\eta_2 \operatorname{Trunc}(\frac{1}{5}\sum_{j = 1}^{5}\epsilon_{t-j}) \in [0,1]$). Again  $\alpha(t)$ is as in the prior simulation. $\sigma^*$ is calculated such that the average variance is equal to 1, e.g., $\bar \sigma = \frac{1}{T}\sum_{t = 1}^{T} E[\Var[Y_{t+1}|I_t = 1, A_t]] = 1$. Note that since $C_t$  is centered in both the model for $I_t$ as well as in the model for $Y_{t+1}$, the standardized proximal main effect is $Z_t'd$ and $E[I_t]=\tau_t=0.5$.  $\alpha_0(t)$ is the conditional mean of $Y_{t+1}$ when $I_t = 0$. The form of $E[Y_{t+1}|I_t = 0]$ is not essential: only $Y_{s+1} - E[Y_{s+1} | I_s =0]$ is used to generate $I_t$. In the simulation, $E[C_t|I_t = 1]$ and $\sigma^*$ are calculated by Monte Carlo methods. As before, the simulated type I error are not affected; see Table \ref{workfalsecase3T} in appendix \ref{sec:A2}. The simulated powers are provided in Table \ref{workfalsecase3P}.

\subsection{Some Practical Guidelines}

	Third, it is critical to use conservative values of $d$ and availability $E[I_t]$ in the sample size formula. It is not surprising that the quality of the sample size formula depends on an accurate or conservative values of the standardized effects, $d$, as this is the case for all sample size formulas.   Additionally availability provides the number of decision points as which treatment might be provided per individual and thus the sample size formula should be sensitive to availability.   To illustrate these points we consider two simulations in which the  data is generated by
	\begin{align*}
	& I_t \stackrel{Ber}{\sim} \big(\tau_t\big), \ \ A_t \stackrel{Ber}{\sim} \big(\rho \big)\\
	& Y_{t+1} = \alpha(t) + (A_t - \rho)Z_t'd + \epsilon_t, \text{ if $I_t = 1$}
	\end{align*}
	where the $\epsilon_t$'s are i.i.d. standard normals and  $\alpha(t)$ is as in the prior simulations. In the first simulation, suppose the scientist provides the correct availability pattern, $\{E[I_t]\}_{t=1}^T$, the correct time at which the maximal standardized proximal main effect is achieved ($\arg\max_t Z_t'd$) and the correct  initial standardized proximal main effect ($Z_1'd=d_1=0$) but provides too low a value of the averaged across time, standardized proximal main effect $\bar d=\frac{1}{T}\sum_{t=1}^{T}Z_t'd$.  The simulated power is provided in Appendix~\ref{sec:A2}, Table~\ref{wrongguessave}.  The degradation in power is pronounced as might be expected.

	In the second simulation, suppose the scientist provides the correct $\arg\max_t Z_t'd$, correct $Z_1'd=d_1=0$, correct $\bar d=\frac{1}{T}\sum_{t=1}^{T}Z_t'd$ and although the scientist's time-varying pattern of availability is correct, the magnitude is underestimated. The simulation result is in Appendix~\ref{sec:A2}, Table~\ref{wrongguessavetau}.  Again the degradation in power is pronounced.
	
\section{Discussion}
	In this paper, we have introduced the use of micro-randomized trials in mobile health and have provided an approach to determining the sample size.  More sophisticated sample size procedures might be entertained.  Certainly it makes sense to include baseline information in the sample size procedure, for example in HeartSteps, a natural baseline variable is baseline step count. The inclusion of baseline variables in $B_t$ in the regression \eqref{lsfit} is straightforward.   An interesting generalization to the sample size procedure would allow scientists to  include time-varying variables (in $S_t$) as covariates in $B_t$ in the regression \eqref{lsfit}.  This might be a useful strategy for reducing the error variance.

An alternate  to the micro-randomized trial design is the single case design often used in the behavioral sciences \citep{DalleryRaiff}.  These trials usually only involve 1 to 13 participants \citep{ShadishSullivan} and the data analyses focus on the examination of  visual trends for each participant separately.   For example, during periods when a participant is on treatment the response might be generally higher than the height of the response during the time periods in which the participant is off treatment.  Dallery et al. \citep{Dallery} provide an excellent overview of single case designs and their use for evaluating technology based intervention. Their paper illustrates the visual analyses that would be conducted on each participant's data. A critical assumption is that the effect of the treatment is only temporary (no carry-over effect) so that each participant can act as his own control.  We believe that in settings in which treatments are expected to have sufficiently strong effects so as to overwhelm the within person variability in response (thus a visual analysis can be compelling), these designs provide an alternative to the micro-randomized trial design.

Although this paper has focused on determining the sample size to detect the proximal main effect of a treatment with a given power,  micro-randomized studies provide data for a variety of interesting further analyses.   For example, it is  of some interest to model and understand the predictors of the time-varying availability indicator.   In the case of HeartSteps we will know why the participant is unavailable (driving a car, already active or has turned off the lock-screen messages) so we will be able to consider each type of availability indicator.   Other very interesting further analyses include assessing interactions between treatments, $A_t$ and context, $S_t$, past treatment $A_s, s<t$ on the proximal response, $Y_{t+1}$.  Also there is much interest in using this type of data to construct \lq\lq dynamic treatment regimes\rq\rq; in this setting these are called Just-in-Time Adaptive Interventions \citep{MetzandNilsen2014}. The sequential micro-randomizations enhance all of these analyses by reducing causal confounding.

\ack
This research was supported by NIH grants P50 DA039838, R01 AA023187, R01HL12544001 and grant U54EB020404 awarded by the National Institute of Biomedical Imaging and Bioengineering (NIBIB) through funds provided by the trans-NIH Big Data to Knowledge (BD2K) initiative (www.bd2k.nih.gov).
AT acknowledges the support of NSF under CAREER grant IIS-1452099. 

\bibliographystyle{stat_in_med}
\bibliography{reference}

\begin{thebibliography}{10}

\bibitem{Lewis2013}
Lewis MA, Uhrig JD, Bann CM, Harris JL, Furberg RD, Coomes C, and Kuhns LM.
\newblock Tailored text messaging intervention for hiv adherence: a
  proof-of-concept study.
\newblock \emph{Health psychology : official journal of the Division of Health
  Psychology, American Psychological Association} 2013;\hspace{0pt}
  \textbf{32}:248---253.

\bibitem{KaplanStone2013}
Kaplan RM and Stone AA.
\newblock Bringing the laboratory and clinic to the community: Mobile
  technologies for health promotion and disease prevention.
\newblock \emph{Annual Review of Psychology} 2013;\hspace{0pt}
  \textbf{64}:471--498.
\newblock PMID: 22994919.

\bibitem{King2013}
King AC, Castro CM, Buman MP, Hekler EB, Urizar J Guido~G, and Ahn DK.
\newblock Behavioral impacts of sequentially versus simultaneously delivered
  dietary plus physical activity interventions: the calm trial.
\newblock \emph{Annals of Behavioral Medicine} 2013;\hspace{0pt}
  \textbf{46}:157--168.

\bibitem{Marsch2012}
Marsch LA.
\newblock Leveraging technology to enhance addiction treatment and recovery.
\newblock \emph{Journal of Addictive Diseases} 2012;\hspace{0pt}
  \textbf{31}:313--318.
\newblock PMID: 22873192.

\bibitem{Boyer2012}
Boyer E, Fletcher R, Fay R, Smelson D, Ziedonis D, and Picard R.
\newblock Preliminary efforts directed toward the detection of craving of
  illicit substances: The iheal project.
\newblock \emph{Journal of Medical Toxicology} 2012;\hspace{0pt}
  \textbf{8}:5--9.

\bibitem{AlssiandPetry2013}
Alessi SM and Petry NM.
\newblock A randomized study of cellphone technology to reinforce alcohol
  abstinence in the natural environment.
\newblock \emph{Addiction} 2013;\hspace{0pt} \textbf{108}:900--909.

\bibitem{Cucciare2012}
A~Cucciare M, R~Weingardt K, J~Greene C, and Hoffman J.
\newblock Current trends in using internet and mobile technology to support the
  treatment of substance use disorders.
\newblock \emph{Current Drug Abuse Reviews} 2012;\hspace{0pt}
  \textbf{5}:172--177.

\bibitem{Gustafson2014}
Gustafson D, FM M, M C, and et~al.
\newblock A smartphone application to support recovery from alcoholism: A
  randomized clinical trial.
\newblock \emph{JAMA Psychiatry} 2014;\hspace{0pt} \textbf{71}:566--572.

\bibitem{Quanbeck2014}
Quanbeck A, Gustafson D, Marsch L, McTavish F, Brown R, Mares ML, Johnson R,
  Glass J, Atwood A, and McDowell H.
\newblock Integrating addiction treatment into primary care using mobile health
  technology: protocol for an implementation research study.
\newblock \emph{Implementation Science} 2014;\hspace{0pt} \textbf{9}:65.

\bibitem{Free2013}
Free C, Phillips G, Galli L, Watson L, Felix L, Edwards P, Patel V, and Haines
  A.
\newblock The effectiveness of mobile-health technology-based health behaviour
  change or disease management interventions for health care consumers: A
  systematic review.
\newblock \emph{PLoS Med} 2013;\hspace{0pt} \textbf{10}:e1001362.

\bibitem{Nilsen2012}
Nilsen W, Kumar S, Shar A, Varoquiers C, Wiley T, Riley WT, Pavel M, and
  Atienza AA.
\newblock Advancing the science of mhealth.
\newblock \emph{Journal of Health Communication} 2012;\hspace{0pt}
  \textbf{17}:5--10.

\bibitem{Muessig2013}
Muessig EK, Pike CE, LeGrand S, and Hightow-Weidman BL.
\newblock Mobile phone applications for the care and prevention of hiv and
  other sexually transmitted diseases: A review.
\newblock \emph{J Med Internet Res} 2013;\hspace{0pt} \textbf{15}:e1.

\bibitem{MetzandNilsen2014}
Spruijt-Metz D and Nilsen W.
\newblock Dynamic models of behavior for just-in-time adaptive interventions.
\newblock \emph{Pervasive Computing, IEEE} 2014;\hspace{0pt}
  \textbf{13}:13--17.

\bibitem{Kumar2013}
Kumar S, Nilsen W, Pavel M, and Srivastava M.
\newblock Mobile health: Revolutionizing healthcare through transdisciplinary
  research.
\newblock \emph{Computer} 2013;\hspace{0pt} \textbf{46}:28--35.

\bibitem{Box1978}
Box GEP, Hunter JS, and Hunter WG.
\newblock \emph{Statistics for experimenters : an introduction to design, data
  analysis, and model building}.
\newblock Wiley series in probability and mathematical statistics, 1978.

\bibitem{Chakraborty2009}
Chakraborty B, Collins LM, Strecher VJ, and Murphy SA.
\newblock Developing multicomponent interventions using fractional factorial
  designs.
\newblock \emph{Statistics in Medicine} 2009;\hspace{0pt}
  \textbf{28}:2687--2708.

\bibitem{Rubin1978}
Rubin DB.
\newblock Bayesian inference for causal effects: The role of randomization.
\newblock \emph{The Annals of Statistics} 1978;\hspace{0pt} \textbf{6}:34--58.

\bibitem{Robins1986}
Robins J.
\newblock A new approach to causal inference in mortality studies with a
  sustained exposure period---application to control of the healthy worker
  survivor effect.
\newblock \emph{Mathematical Modelling} 1986;\hspace{0pt} \textbf{7}:1393 --
  1512.

\bibitem{Robins1987}
Robins J.
\newblock Addendum to ``a new approach to causal inference in mortality studies
  with a sustained exposure period---application to control of the healthy
  worker survivor effect''.
\newblock \emph{Computers and Mathematics with Applications} 1987;\hspace{0pt}
  \textbf{14}:923 -- 945.

\bibitem{Wang2012}
Wang L, Rotnitzky A, Lin X, Millikan RE, and Thall PF.
\newblock Evaluation of viable dynamic treatment regimes in a sequentially
  randomized trial of advanced prostate cancer.
\newblock \emph{Journal of the American Statistical Association}
  2012;\hspace{0pt} \textbf{107}:493--508.

\bibitem{Robins2004}
Robins JM.
\newblock Optimal structural nested models for optimal sequential decisions.
\newblock \emph{In Proceedings of the Second Seattle Symposium on
  Biostatistics} 2004;\hspace{0pt} \textbf{179}:189--326.

\bibitem{GEE}
Liang KY and Zeger SL.
\newblock Longitudinal data analysis using generalized linear models.
\newblock \emph{Biometrika} 1986;\hspace{0pt} \textbf{73}:13--22.

\bibitem{Tu2004}
Tu XM, Kowalski J, Zhang J, Lynch KG, and Crits-Christoph P.
\newblock Power analyses for longitudinal trials and other clustered designs.
\newblock \emph{Statistics in Medicine} 2004;\hspace{0pt}
  \textbf{23}:2799--2815.

\bibitem{ManclandDeRouen2001}
Mancl LA and DeRouen TA.
\newblock A covariance estimator for {GEE} with improved small-sample
  properties.
\newblock \emph{Biometrics} 2001;\hspace{0pt} \textbf{57}:126--134.

\bibitem{LiRedden}
Li P and Redden DT.
\newblock Small sample performance of bias-corrected sandwich estimators for
  cluster-randomized trials with binary outcomes.
\newblock \emph{Statistics in Medicine} 2015;\hspace{0pt} \textbf{34}:281--296.

\bibitem{Hotelling}
Hotelling H.
\newblock The generalization of student's ratio.
\newblock \emph{Ann Math Statist} 1931;\hspace{0pt} \textbf{2}:360--378.

\bibitem{Cohen1988}
Cohen J.
\newblock \emph{Statistical Power Analysis for the Behavioral Sciences(2nd)}.
\newblock Routledge, 2nd edition, 1988.

\bibitem{DalleryRaiff}
Dallery J and Raiff B.
\newblock Optimizing behavioral health interventions with single-case designs:
  from development to dissemination.
\newblock \emph{Translational Behavioral Medicine} 2014;\hspace{0pt}
  \textbf{4}:290--303.

\bibitem{ShadishSullivan}
Shadish W and Sullivan K.
\newblock Characteristics of single-case designs used to assess intervention
  effects in 2008.
\newblock \emph{Behavior Research Methods} 2011;\hspace{0pt}
  \textbf{43}:971--980.

\bibitem{Dallery}
Dallery J, Cassidy R, and Raiff B.
\newblock Single-case experimental designs to evaluate novel technology-based
  health interventions.
\newblock \emph{Journal of Medical Internet Research} 2013;\hspace{0pt}
  \textbf{15}:e22.

\end{thebibliography}

\newpage
\begin{appendices}
	
	\section{Theoretical Results and Proofs}
	\label{sec:A1}
	
	\begin{lemma}[Least Squares Estimator]
		\label{lse}
		The least square estimators $\hat \alpha, \hat \beta$ are consistent estimators of $\tilde \alpha, \tilde \beta$ in \eqref{tilde.alpha} and \eqref{tilde.beta}. In particular, if $\beta(t) = Z_t'\beta^*$ for some vector $\beta^*$, then $\tilde{\beta} = \beta^*$.  Under moment conditions, we have $
		\sqrt{N}(\hat \beta - \tilde \beta) \rightarrow N(0, \Sigma_\beta)$,
		where the asymptotic variance $\Sigma_\beta$ is given by $
		\Sigma_\beta = Q^{-1}WQ^{-1}
		$
		where $Q=\sum_{t=1}^TE[I_t]\rho_t(1-\rho_t)Z_tZ_t'$, $
		W=  E\bigg[\sum_{t=1}^T \tilde{\epsilon}_tI_t(A_t -\rho_t)Z_t\times\sum_{t=1}^T\tilde{\epsilon}_tI_t(A_t -\rho_t)Z_t'\bigg]$ and $\tilde \epsilon_t = Y_{t+1} - B_t'\tilde{\alpha} - Z_t'\tilde{\beta}(A_t - \rho_t)$.
	\end{lemma}
	\begin{proof}
		It's easy to see that the least square estimators satisfy
		\begin{align*}
			\hat \theta = (\hat \alpha, \hat \beta) &= \bigg(\mathbb{P}_N\sum_{t=1}^T I_tX_tX_t' \bigg)^{-1} \bigg(\mathbb{P}_N \sum_{t=1}^T I_tY_{t+1}X_t\bigg)\label{thetahat}\\
			& \rightarrow \bigg(\sum_{t=1}^T E(I_tX_tX_t') \bigg)^{-1} \bigg(\sum_{t=1}^T E(I_tY_{t+1}X_t)\bigg)
		\end{align*}
		where $X_t' = (B_t', (A_t - \rho_t)Z_t')\in \mathbb{R}^{1\times (p+q)}$ is the covariate at time t. For each t,
		\begin{align*}
			& E(I_tX_tX_t') =
			\begin{pmatrix}
				E[I_t] B_tB_t'   & B_tZ_t'E[I_t(A_t - \rho_t)]\\
				Z_tB_t'E[I_t(A_t - \rho_t)] & Z_tZ_t'E[I_t(A_t - \rho_t)^2]
			\end{pmatrix} =
			\begin{pmatrix}
				E[I_t] B_tB_t' & 0\\
				0 & E[I_t] \rho_t(1-\rho_t)Z_tZ_t'
			\end{pmatrix}\\
			& E(I_tY_{t+1}X_t) =
			\begin{pmatrix}
				E[I_tY_{t+1}]B_t \\
				E[I_tY_{t+1}(A_t - \rho_t)]Z_t
			\end{pmatrix} =
			\begin{pmatrix}
				E[I_tY_{t+1}]B_t \\
				\rho_t(1-\rho_t)E[I_t]\beta(t)Z_t
			\end{pmatrix},
		\end{align*}
		so that
		\[
		\hat\alpha \rightarrow \left(\sum_{t=1}^T E[I_t]B_tB_t'\right)^{-1}\sum_{t=1}^TE[I_tY_{t+1}]B_t = \left(\sum_{t=1}^T E[I_t]B_tB_t'\right)^{-1}\sum_{t=1}^TE[I_t]\alpha(t)B_t
		\]
		\[
		\hat\beta \rightarrow  \left(\sum_{t=1}^T \rho_t(1-\rho_t) E[I_t]Z_tZ_t'\right)^{-1}\sum_{t=1}^TE [I_tY_{t+1}(A_t - \rho_t)]Z_t
		= \left(\sum_{t=1}^T \rho_t(1-\rho_t) E[I_t]Z_tZ_t'\right)^{-1}\sum_{t=1}^T E[I_t]\rho_t(1-\rho_t)\beta(t) Z_t
		\]
		as in \eqref{tilde.alpha} and \eqref{tilde.beta}. We can see that if $\beta(t) = Z_t'\beta^*$, then $\left(\sum_{t=1}^T \rho_t(1-\rho_t) E[I_t]Z_tZ_t'\right)^{-1}\sum_{t=1}^T E[I_t]\rho_t(1-\rho_t)\beta(t) Z_t = \left(\sum_{t=1}^T \rho_t(1-\rho_t) E[I_t]Z_tZ_t'\right)^{-1}\sum_{t=1}^T E[I_t]\rho_t(1-\rho_t)Z_tZ_t'\beta^*= \beta^*$. This is true even if $E[Y_{t+1}|I_t = 1] \neq B_t'\tilde{\alpha}$.
		
		We can easily see that,
		\begin{align}
			\sqrt{N}(\hat \theta - \tilde \theta) & = \sqrt{N}\Bigg\{\big(\mathbb{P}_N\sum_{t=1}^T I_tX_tX_t' \big)^{-1} \Big[\big(\mathbb{P}_N \sum_{t=1}^T I_tY_{t+1}X_t\big)- \big(\mathbb{P}_N\sum_{t=1}^T I_tX_tX_t'\big)\tilde{\theta}\Big]\Bigg\} \notag \\
			& = \sqrt{N}\bigg\{E\big[\sum_{t=1}^T I_tX_tX_t' \big]^{-1}\big(\mathbb{P}_N \sum_{t=1}^T I_t\tilde{\epsilon}_tX_t\big)   \bigg\} + o_p(\mathbf{1}), \label{thetahat}
		\end{align}
		where $o_p(\mathbf{1})$ is a term that converges in probability to zero as $N$ goes to infinity. By the definitions of $\tilde{\alpha}$ and $\tilde \beta$, we have
		\begin{align*}
			E\big[\sum_{t=1}^T I_t\tilde{\epsilon}_tX_t\big] =  \begin{pmatrix}
				\sum_{t=1}^T E[I_t]\left(\alpha(t) - B_t'\tilde{\alpha}\right)B_t \\
				\sum_{t=1}^T E[I_t] \rho_t(1-\rho_t)\big(\beta(t) - Z_t'\tilde{\beta}\big) Z_t
			\end{pmatrix} = \mathbf{0}
		\end{align*}
		So that under moments conditions, we have $\sqrt{N}(\hat \theta - \tilde \theta) \rightarrow N(0, \Sigma_\theta )$, where $\Sigma_\theta$ is given by $$\Sigma_\theta = E\big[\sum_{t=1}^T I_tX_tX_t' \big]^{-1} E\big[\sum_{t=1}^T I_t\tilde{\epsilon}_tX_t \times \sum_{t=1}^T I_t\tilde{\epsilon}_tX_t'\big] E\big[\sum_{t=1}^T I_tX_tX_t' \big]^{-1} = \begin{bmatrix}
		\Sigma_\alpha &\Sigma_{\alpha\beta} \\
		\Sigma_{\alpha\beta}' &\Sigma_\beta
		\end{bmatrix}.$$ In particular, $\hat \beta$ satisfies $\sqrt{N}(\hat \beta - \tilde{\beta}) \rightarrow N(0, \Sigma_\beta)$ and $\Sigma_\beta$ is given by
		\begin{align*}
			\Sigma_\beta &= \bigg(\sum_{t=1}^T E[I_t]\rho_t(1-\rho_t) Z_tZ_t' \bigg)^{-1} E\bigg[\sum_{t=1}^T \tilde{\epsilon}_tI_t(A_t -\rho_t)Z_t\times\sum_{t=1}^T\tilde{\epsilon}_tI_t(A_t -\rho_t)Z_t'\bigg]  \bigg(\sum_{t=1}^T E[I_t] \rho_t(1-\rho_t) Z_tZ_t' \bigg)^{-1} = Q^{-1}WQ^{-1}.
		\end{align*}
	\end{proof}
	
	\begin{lemma}[Asymptotic Variance Under Working Assumptions]
		\label{simpli}
		Assuming working assumptions (\ref{assum:a})-(\ref{assum:d}) are true, then under the alternative hypothesis $H_1$ in (\ref{alternative2}), $\Sigma_{\beta}$ and $c_N$ are given by
		\begin{equation*}
			\Sigma_\beta  = \bar\sigma^2\bigg( \sum_{t=1}^{T}E[I_t]\rho_t(1-\rho_t)Z_tZ_t'\bigg)^{-1},
		\end{equation*}
		\begin{equation*}
			c_N = N{d}'\bigg( \sum_{t=1}^{T}E[I_t]\rho_t(1-\rho_t)Z_tZ_t'\bigg) {d}.
		\end{equation*}
		
	\end{lemma}
	
	\begin{proof}
		Note that under assumptions (\ref{assum:b}) and (\ref{assum:c}), we have $Z_t'\tilde{\beta} = \beta(t)$ and $\Var(Y_{t+1}|I_t=1,A_t) = \bar \sigma$ for each t, and $\tilde{d} = d$. The middle term, $W$, in $\Sigma_\beta$ can be separated by two terms, e.g., $E\bigg[\sum_{t=1}^T \tilde{\epsilon}_tI_t(A_t -\rho_t)Z_t\times\sum_{t=1}^T\tilde{\epsilon}_tI_t(A_t -\rho_t)Z_t'\bigg] = \sum_{t=1}^T E\big[ \tilde{\epsilon}_t^2I_t(A_t -\rho_t)^2\big]Z_tZ_t' + \sum_{i\neq j}^TE\big[ \tilde{\epsilon}_i\tilde{\epsilon}_jI_iI_j(A_i -\rho_i)(A_j - \rho_j)\big]Z_iZ_j'$. Under assumptions (\ref{assum:a}), (\ref{assum:b}) and (\ref{assum:c}), we have $E[\tilde{\epsilon}_t|I_t = 1, A_t] = 0$ and $E\big[ \tilde{\epsilon}_t^2I_t(A_t -\rho_t)^2\big] = E[I_t]\rho_t(1-\rho_t)\bar \sigma^2$. Furthermore, suppose $i>j$, then $E\big[ \tilde{\epsilon}_i\tilde{\epsilon}_jI_iI_j(A_i -\rho)(A_j - \rho)\big] = E[I_iI_j(A_j-\rho)(A_i-\rho)] \times E[\tilde\epsilon_t\tilde\epsilon_s | I_t=1, I_s = 1, A_t, A_s] = 0$, because $A_i \indep \{I_i, I_j, A_j\}$ and the first term is 0. $W$ is then given by $$W = \bar \sigma^2 \sum_{t=1}^T E[I_t]\rho_t(1-\rho_t)Z_tZ_t',$$
		so that $\Sigma_\beta  = \bar\sigma^2\big( \sum_{t=1}^{T}E[I_t]\rho_t(1-\rho_t)Z_tZ_t'\big)^{-1}$ and $c_N = N(\bar\sigma \tilde{d})^{'}\Sigma_\beta^{-1}(\bar\sigma \tilde{d}) = N{d}'\bigg( \sum_{t=1}^{T}E[I_t]\rho_t(1-\rho_t)Z_tZ_t'\bigg) {d}$.
	\end{proof}
	\textit{Remark:} Working assumption (d) can be replaced by assuming $E[Y_{t+1}|I_t = 1, A_t, I_s = 1, A_s] - E[Y_{t+1}|I_t = 1, A_t]$ does not depend on $A_t$ for any $s < t$, or some Markovian type of assumption, e.g., $Y_{t+1} \indep \{ Y_{s+1}, I_s, A_s, s < t\}|I_t, A_t$. Either of them implies  $E\big[ \tilde{\epsilon}_i\tilde{\epsilon}_jI_iI_j(A_i -\rho_i)(A_j - \rho_j)\big]=0$, so that $\Sigma_\beta$ and $c_N$ have the same simplified forms.

	\paragraph*{Rationale for multiple of F distribution}
	The distribution of the quadratic form, $n(\bar X-\mu)'\hat\Sigma^{-1}(\bar X-\mu)$ constructed from a random sample of size $n$ of N($\mu,\Sigma$) random variables in which $\hat\Sigma$ is the sample covariance matrix follows a   Hotelling's $T$-squared distribution.  The Hotelling's $T$-squared distribution is a  multiple of the F distribution, $\frac{d_1(d_1+d_2-1)}{d_2} F_{d_1, d_2}$  in which $d_1$ is the dimension of $\mu$, and $d_2$ is the sample size.   Our sample sample approximation replaces $d_1$ by $p$ (the number of parameters in the test statistic) and $d_2$ by $n-q-p$ (the sample size minus the number of nuisance parameters minus $d_1$).
	
	\paragraph{{Formula for adjusted $\hat W$ and $\hat Q$}}

	Define a individual-specific residual vector $\hat e$ as the $T\times 1$ vector with $t$th entry $\hat e_t=Y_{t+1}-  I_t
	B_t'\hat\alpha - I_t(A_t -\rho_t)Z_t'\hat\beta$. For each individual define the $t$th row of the $T\times (p+q)$ individual-specific matrix $X$ by $(I_tB_t', I_t(A_t-\rho_t)Z_t)$.  Then define $H=X \left[ \mathbb{P}_NX'X\right]^{-1}X'$.
	The matrix $\hat Q^{-1}$ is given by the lower right $p\times p$ block in the inverse of $\left[ \mathbb{P}_NX'X\right]$; the matrix $\hat W$ is given by the lower right $p\times p$ block in $  \mathbb{P}_N\left[ X^{T}(I-H)^{-1}\hat{e}\hat{e}'(I-H)^{-1}X\right]$.\\
	
\end{appendices}


\newpage

\setcounter{table}{0}
\renewcommand{\thetable}{\arabic{table}B}
\setcounter{section}{1}

\begin{appendices}
\section{Further Simulations and Details}
\label{sec:A2}

\subsection{Simulation Results When Working Assumptions are True}
\label{sec:B1}
We conduct a variety of simulations in settings in which the working assumptions hold, the scientist provides the correct pattern for the expected availability, $\tau_t = E[I_t]$ and under the alternate,  the standardized proximal main effect is $d(t) = Z_t'd$. Here we will mainly focus on the setup where the duration of the study is 42 days and there are 5 decision times within each day, but similar results can be obtained in different setups; see below. The randomization probability is 0.4, i.e. $\rho = \rho_t= P(A_t = 1) =  0.4$. The sample size formula is given in (\ref{cN}) and (\ref{samplesizeeq}).  The test statistic is given by (\ref{teststat}) in which $B_t$ and $Z_t$ equal to $\left(1, \lfloor \frac{t-1}{5}\rfloor,\lfloor \frac{t-1}{5}\rfloor^2 \right)'$. All simulations are based on 1,000 simulated data sets. The significance level is 0.05 and the desired power is 80\%.

In the first simulation,  the data for each simulated subject is generated sequentially as follows. For $t = 1,\dots, T=210$, $I_t$, $A_t$ and $Y_{t+1}$ are generated by
\begin{align*}
& I_t \stackrel{Ber}{\sim} \big(\tau_t\big), \ \ A_t \stackrel{Ber}{\sim} \big(\rho \big)\\
& Y_{t+1} = \alpha(t) + (A_t - \rho)d(t) + \epsilon_t, \text{ if $I_t = 1$}
\end{align*}
where $d(t) = Z_t'd$ and $\tau_t$ are same as in the sample size model. The conditional mean, $E[Y_{t+1}|I_t=1]=\alpha(t)$ is given by $\alpha(t) = \alpha_1 + \alpha_2 \lfloor \frac{t-1}{5} \rfloor + \alpha_3 \lfloor \frac{t-1}{5} \rfloor ^2$, where $\alpha_1 = 2.5$, $\alpha_2 = 0.727$,$\alpha_3 = - 8.66\times 10^{-4}$ (so that $(1/T)\sum_t\alpha(t) - \alpha(1)= 1$, $\operatorname{arg\,max}_t\alpha(t) = T$). We consider 5 differing distributions for the errors $\{\epsilon_t\}_{t = 1}^{T}$: independent normal; independent (scaled) Student's $t$ distribution with 3 degrees of freedom; independent (centered) exponential distribution with $\lambda = 1$;   a Gaussian AR(1) process, e.g., $\epsilon_t = \phi\epsilon_{t-1} + v_t$, where $v_t$ is white noise with variance $\sigma_v^2$ such that $\Var(\epsilon_t) = 1$; and lastly a Gaussian AR(5) process, e.g., $\epsilon_t = \frac{\phi}{5}\sum_{j = 1}^{5}\epsilon_{t-j} + v_t$, where $v_t$ is white noise with variance $\sigma_v^2$ such that $\Var(\epsilon_t) = 1$.   In all cases the errors are scaled to have mean 0 and variance 1 (i.e. $E[\epsilon_t |I_t = 1]=0$, $\Var[\epsilon_t |A_t, I_t = 1]=1$). Additionally four availability patterns, e.g., time varying values for $\tau_t=E[I_t]$, are considered; see Figure~(\ref{ShapeOfTau}). The simulated type 1 error rate and power when the duration of study is 42 days are reported in Table \ref{worktrue42T} and \ref{worktrue42P}. The simulation results in other setups, e.g., the length of the study is 4 week and 8 week, are reported in Table \ref{worktrue2856}. The associated sample sizes are given in Table \ref{worktrueformula}.

Since neither the working assumptions nor the inputs to the sample size formula specify the dependence of the availability indicator, $I_t$ on past treatment. In the second simulation, we consider  the setting in which the availability decreases as the number of treatments provided in the recent past increase. In particular,  the data are generated as follows,
\begin{align*}
& I_t \stackrel{Ber}{\sim} \big(\tau_t + \eta \sum_{j=1}^{5} (A_{t-j}I_{t-j} - E[A_{t-j}I_{t-j}])\big), \ \ A_t \stackrel{Ber}{\sim} \big(\rho \big)\\
& Y_{t+1} = \alpha(t) + (A_t - \rho)d(t) + \epsilon_t, \text{ if $I_t = 1$}
\end{align*}
Note that since we center $\sum_{j=1}^{5} A_{t-j}I_{t-j}$ in the generative model of $I_t$, the expected availability is $\tau_t$. The specification of $\alpha(t)$, $\beta(t)$ and $\epsilon_t$ are same as in the first simulation. The simulated type I error rate and power are reported Table \ref{worktruedep}.

\subsection{Further Details When Working Assumptions are False }

\subsubsection{Working Assumption (a) is Violated.}
\label{sec:B2}
Here we consider another setting in which the working assumption (a) is violated, e.g., the underlying true $E[Y_{t+1}|I_t = 1]$ follows a non-quadratic form (recall that $B_t$ is given by $\left(1, \lfloor \frac{t-1}{5}\rfloor,\lfloor \frac{t-1}{5}\rfloor^2 \right)'$). The data is generated as follows
\begin{align*}
& I_t \stackrel{Ber}{\sim} \big(\tau_t\big), \ \ A_t \stackrel{Ber}{\sim} \big(\rho \big)\\
& Y_{t+1} = \alpha(t) +  (A_t - \rho)Z_t'd + \epsilon_t, \text{ if $I_t = 1$}
\end{align*}
where $\alpha(t)=E[Y_{t+1}|I_t = 1] $ is provided in Figure \ref{ShapeOfAlpha}. For each case, $\alpha(t)$ satisfies $\alpha(1) = 2.5$ and $(1/T)\sum_{t=1}^{T} - \alpha(1) = 0.1$. The error terms $\{\epsilon_t\}_{t=1}^N$ are i.i.d N$(0, 1)$. The day of maximal proximal effect is assumed to be 29.  Additionally, different values of averaged standardized treatment effect and four patterns of availability in Figure \ref{ShapeOfTau} with average 0.5 are considered. The simulation results are reported in Table \ref{workfalsea2}.

\subsubsection{Additional Simulation Results When Other Working Assumptions are False}
The main body of the paper reports part of the results when working assumptions (b), (c) and (d) are violated. Additional simulation results are provided here. In particular,  the simulation result is reported in Table \ref{workfalsecase1full} when $d(t)$ follows other non-quadratic forms, e.g., working assumption (b) is false; see Figure \ref{ShapeOfBetafull}. The simulated Type I error rate and power when working assumption (c) is false are reported in Table \ref{workfalsecase2T}. The simulated Type I error rate when  working assumption (d) is violated is reported in Table \ref{workfalsecase3T}.

\subsubsection{Simulation Results when $\bar d$ and $\bar \tau$ are misspecified.}
As discussed in the paper, the first scenario considers the setting in which the scientist provides the correct availability pattern, $\{E[I_t]\}_{t=1}^T$, the correct time at which the maximal standardized proximal main effect is achieved ($\arg\max_t Z_t'd$) and the correct  initial standardized proximal main effect ($Z_1'd=d_1=0$) but provides too low a value of the averaged across time, standardized proximal main effect $\bar d=\frac{1}{T}\sum_{t=1}^{T}Z_t'd$.  The simulated power is provided in Table~\ref{wrongguessave}. In the second scenario, the scientist provides the correct $\arg\max_t Z_t'd$, correct $Z_1'd=d_1=0$, correct $\bar d=\frac{1}{T}\sum_{t=1}^{T}Z_t'd$ and although the scientist's time-varying pattern of availability is correct, the magnitude, e.g., the average availability,  is underestimated. The simulation result is in Table~\ref{wrongguessavetau}.

\begin{table}[H]
	\centering
	\begin{threeparttable}
		\caption{Sample Sizes when the proximal treatment effect satisfies $d(t) = Z_t'd$. The significance level is 0.05. The desired power is 0.80.}
		\begin{tabular}{|c|c|c|ccc|ccc|}

			\multirow{3}{*}{Duration of Study} & \multirow{3}{*}{Availability Pattern} & \multirow{3}{*}{Max} &      & $\bar \tau$ = 0.5 &      &      & $\bar \tau$= 0.7 &  \\ \cline{4-9}
			                                   &                                       &                      &          \multicolumn{6}{c|}{Average Proximal Effect }           \\ \cline{4-9}
			                                   &                                       &                      & 0.10 &       0.08        & 0.06 & 0.10 &       0.08       & 0.06 \\ \hline
			     \multirow{12}{*}{4-week}      &      \multirow{3}{*}{Pattern 1}       &          15          &  59  &        89         & 154  &  43  &        65        & 112  \\
			                                   &                                       &          22          &  60  &        91         & 158  &  44  &        66        & 114  \\
			                                   &                                       &          29          &  58  &        87         & 152  &  43  &        64        & 110  \\ \cline{2-9}
			                                   &      \multirow{3}{*}{Pattern 2}       &          15          &  59  &        89         & 154  &  43  &        65        & 112  \\
			                                   &                                       &          22          &  60  &        92         & 159  &  44  &        67        & 115  \\
			                                   &                                       &          29          &  58  &        89         & 154  &  43  &        64        & 111  \\ \cline{2-9}
			                                   &      \multirow{3}{*}{Pattern 3}       &          15          &  59  &        90         & 157  &  44  &        66        & 113  \\
			                                   &                                       &          22          &  63  &        96         & 167  &  46  &        69        & 119  \\
			                                   &                                       &          29          &  62  &        94         & 163  &  45  &        67        & 115  \\ \cline{2-9}
			                                   &      \multirow{3}{*}{Pattern 4}       &          15          &  59  &        89         & 155  &  43  &        65        & 112  \\
			                                   &                                       &          22          &  57  &        86         & 150  &  43  &        64        & 110  \\
			                                   &                                       &          29          &  54  &        82         & 142  &  41  &        61        & 105  \\ \hline
			     \multirow{12}{*}{6-week}      &      \multirow{3}{*}{Pattern 1}       &          22          &  41  &        61         & 105  &  31  &        45        &  76  \\
			                                   &                                       &          29          &  42  &        64         & 109  &  32  &        47        &  79  \\
			                                   &                                       &          36          &  41  &        62         & 106  &  31  &        45        &  77  \\ \cline{2-9}
			                                   &      \multirow{3}{*}{Pattern 2}       &          22          &  41  &        61         & 105  &  31  &        45        &  76  \\
			                                   &                                       &          29          &  43  &        64         & 110  &  32  &        47        &  80  \\
			                                   &                                       &          36          &  42  &        62         & 107  &  31  &        46        &  77  \\ \cline{2-9}
			                                   &      \multirow{3}{*}{Pattern 3}       &          22          &  42  &        62         & 106  &  31  &        46        &  77  \\
			                                   &                                       &          29          &  44  &        66         & 114  &  33  &        48        &  82  \\
			                                   &                                       &          36          &  43  &        65         & 112  &  32  &        47        &  80  \\ \cline{2-9}
			                                   &      \multirow{3}{*}{Pattern 4}       &          22          &  41  &        62         & 106  &  31  &        45        &  77  \\
			                                   &                                       &          29          &  41  &        62         & 106  &  31  &        46        &  78  \\
			                                   &                                       &          36          &  40  &        59         & 101  &  30  &        44        &  74  \\ \hline
			     \multirow{12}{*}{8-week}      &      \multirow{3}{*}{Pattern 1}       &          29          &  32  &        47         &  80  &  25  &        35        &  58  \\
			                                   &                                       &          36          &  33  &        49         &  84  &  26  &        37        &  61  \\
			                                   &                                       &          43          &  33  &        48         &  82  &  25  &        36        &  60  \\ \cline{2-9}
			                                   &      \multirow{3}{*}{Pattern 2}       &          29          &  32  &        47         &  80  &  25  &        35        &  58  \\
			                                   &                                       &          36          &  34  &        49         &  84  &  26  &        37        &  61  \\
			                                   &                                       &          43          &  33  &        49         &  82  &  25  &        36        &  60  \\ \cline{2-9}
			                                   &      \multirow{3}{*}{Pattern 3}       &          29          &  33  &        48         &  82  &  25  &        36        &  59  \\
			                                   &                                       &          36          &  35  &        51         &  87  &  26  &        38        &  63  \\
			                                   &                                       &          43          &  34  &        50         &  86  &  26  &        37        &  62  \\ \cline{2-9}
			                                   &      \multirow{3}{*}{Pattern 4}       &          29          &  33  &        48         &  81  &  25  &        36        &  59  \\
			                                   &                                       &          36          &  33  &        49         &  83  &  25  &        36        &  61  \\
			                                   &                                       &          43          &  32  &        47         &  80  &  25  &        35        &  59  \\ \hline
		\end{tabular}
		\begin{tablenotes}
			\item   \lq\lq Max\rq\rq is the day in which the maximal proximal effect is attained.   $\bar\tau = (1/T)\sum_{t=1}^{T}E[I_t]$ is the average availability.
		\end{tablenotes}
		\label{worktrueformula}
	\end{threeparttable}
	
\end{table}

\newpage
\begin{table}[H]
	
	\centering
	\begin{threeparttable}
		
		\caption{Simulated Type I error rate ($\%$) when working assumptions are true. Duration of the study is 6-week. The associated sample size is given in Table \ref{worktrueformula}.}
		\begin{tabular}{|c|c|c|ccc||ccc|}

			\multirow{3}{*}{Error Term}      & \multirow{3}{*}{Availability Pattern} & \multirow{3}{*}{Max} &      & $\bar \tau$ = 0.5 &      &      & $\bar \tau$= 0.7 &  \\ \cline{4-9}
			&                                       &                      &          \multicolumn{6}{c|}{Average Proximal Effect }           \\ \cline{4-9}
			&                                       &                      & 0.10 &       0.08        & 0.06 & 0.10 &       0.08       & 0.06 \\ \hline
			\multirow{12}{*}{i.i.d. Normal}    &      \multirow{3}{*}{Pattern 1}       &          22          & 3.8  &        4.5        & 4.9  & 4.6  &       5.3        & 4.8  \\
			&                                       &          29          & 4.7  &        6.0        & 4.6  & 4.0  &       3.2        & 5.0  \\
			&                                       &          36          & 5.0  &        5.4        & 4.9  & 4.3  &       4.8        & 4.6  \\ \cline{2-9}
			&      \multirow{3}{*}{Pattern 2}       &          22          & 4.8  &        4.1        & 4.8  & 4.4  &       3.5        & 4.1  \\
			&                                       &          29          & 4.3  &        6.2        & 3.2  & 4.6  &       4.2        & 4.2  \\
			&                                       &          36          & 4.5  &        4.8        & 5.2  & 4.5  &       3.5        & 5.4  \\ \cline{2-9}
			&      \multirow{3}{*}{Pattern 3}       &          22          & 4.7  &        4.5        & 6.3  & 4.4  &       4.9        & 4.9  \\
			&                                       &          29          & 4.1  &        5.1        & 4.6  & 4.3  &       6.0        & 5.6  \\
			&                                       &          36          & 4.7  &        4.4        & 4.6  & 4.1  &       5.1        & 4.4  \\ \cline{2-9}
			&      \multirow{3}{*}{Pattern 4}       &          22          & 5.4  &        3.5        & 4.5  & 4.8  &       4.7        & 5.0  \\
			&                                       &          29          & 5.2  &        4.5        & 4.5  & 5.0  &       5.0        & 5.1  \\
			&                                       &          36          & 3.8  &        4.1        & 5.4  & 4.7  &       5.0        & 5.9  \\ \hline
			\multirow{3}{*}{i.i.d. t dist.}    &      \multirow{3}{*}{Pattern 1}       &          22          & 4.3  &        4.4        & 3.2  & 4.1  &       4.1        & 5.2  \\
			&                                       &          29          & 5.0  &        3.8        & 3.2  & 3.7  &       4.2        & 6.3  \\
			&                                       &          36          & 4.3  &        4.5        & 4.0  & 5.0  &       5.7        & 5.4  \\ \hline
			\multirow{3}{*}{i.i.d. Exp.}      &      \multirow{3}{*}{Pattern 1}       &          22          & 4.5  &        4.6        & 4.4  & 3.7  &       \textbf{7.1}        & 3.1  \\
			&                                       &          29          & 4.5  &        4.6        & 4.2  & 4.5  &       4.5        & 4.7  \\
			&                                       &          36          & 2.7  &        4.8        & 4.8  & 3.9  &       3.7        & 3.4  \\ \hline
			\multirow{3}{*}{AR(1), $\phi = -0.6$} &      \multirow{3}{*}{Pattern 1}       &          22          & 4.3  &        5.3        & 4.6  & 3.8  &       4.2        & 4.0  \\
			&                                       &          29          & 4.6  &        5.4        & 5.1  & 4.0  &       4.4        & 4.3  \\
			&                                       &          36          & 4.7  &        4.0        & 4.0  & 4.1  &       4.2        & 3.9  \\ \hline
			\multirow{3}{*}{AR(1), $\phi = -0.3$} &      \multirow{3}{*}{Pattern 1}       &          22          & 5.8  &        3.4        & 4.4  & 3.3  &       4.0        & 5.4  \\
			&                                       &          29          & 4.9  &        4.7        & 4.6  & 5.5  &       5.5        & 4.5  \\
			&                                       &          36          & 4.0  &        4.7        & 4.4  & 4.9  &       5.0        & 4.7  \\ \hline
			\multirow{3}{*}{AR(1), $\phi = 0.3$}  &      \multirow{3}{*}{Pattern 1}       &          22          & 4.6  &        4.6        & 4.9  & 4.3  &       5.4        & 4.1  \\
			&                                       &          29          & 4.8  &        5.3        & 4.1  & 4.3  &       4.2        & 5.2  \\
			&                                       &          36          & 3.6  &        3.9        & 4.9  & 4.8  &       4.9        & 4.9  \\ \hline
			\multirow{3}{*}{AR(1), $\phi = 0.6$}  &      \multirow{3}{*}{Pattern 1}       &          22          & 4.4  &        5.1        & 4.9  & 3.6  &       5.2        & 3.7  \\
			&                                       &          29          & 3.7  &        4.9        & 4.6  & 4.5  &       4.3        & 5.8  \\
			&                                       &          36          & 4.4  &        \textbf{6.7}        & 5.2  & 5.6  &       3.6        & 5.1  \\ \hline
			\multirow{3}{*}{AR(5), $\phi = -0.6$} &      \multirow{3}{*}{Pattern 1}       &          22          & 4.4  &        4.7        & 5.1  & 4.2  &       4.5        & 5.5  \\
			&                                       &          29          & 4.3  &        5.1        & 4.3  & 3.2  &       3.5        & 4.2  \\
			&                                       &          36          & 5.3  &        4.5        & 6.1  & 4.2  &       4.6        & 5.4  \\ \hline
			\multirow{3}{*}{AR(5), $\phi = -0.3$} &      \multirow{3}{*}{Pattern 1}       &          22          & 3.7  &        4.4        & 6.0  & 5.0  &       4.5        & 3.5  \\
			&                                       &          29          & 4.4  &        4.7        & 5.2  & 5.3  &       4.5        & 5.0  \\
			&                                       &          36          & 4.5  &        5.0        & 5.1  & 4.1  &       5.3        & 4.8  \\ \hline
			\multirow{3}{*}{AR(5), $\phi = 0.3$}  &      \multirow{3}{*}{Pattern 1}       &          22          & 5.3  &        4.3        & 5.7  & 4.8  &       4.1        & 4.3  \\
			&                                       &          29          & 3.9  &        4.8        & 4.1  & 4.0  &       4.3        & 4.9  \\
			&                                       &          36          & 4.2  &        5.5        & 5.1  & 3.6  &       4.5        & 3.6  \\ \hline
			\multirow{3}{*}{AR(5), $\phi = 0.6$}  &      \multirow{3}{*}{Pattern 1}       &          22          & 5.1  &        4.5        & 4.0  & 4.5  &       3.8        & 5.2  \\
			&                                       &          29          & 5.2  &        4.8        & 4.5  & 2.9  &       5.3        & 4.4  \\
			&                                       &          36          & 4.1  &        3.6        & 4.6  & 3.9  &       4.4        & 4.9  \\ \hline
		\end{tabular}
		\begin{tablenotes}
			\item   \lq\lq Max\rq\rq is the day in which the maximal proximal effect is attained.   $\bar\tau = (1/T)\sum_{t=1}^{T}E[I_t]$ is the average availability. $\phi$ is the parameter for AR(1) and AR(5) process.  Bold numbers are significantly(at .05 level) greater than .05.
		\end{tablenotes}
		\label{worktrue42T}		
	\end{threeparttable}
\end{table}

\begin{table}[H]
	\centering
	\begin{threeparttable}	
		\caption{Simulated power($\%$) when working assumptions are true. Duration of the study is 6-week. The associated sample size is given in Table \ref{worktrueformula} }
		\begin{tabular}{|c|c|c|ccc|ccc|}

			\multirow{3}{*}{Error Term}      & \multirow{3}{*}{Availability Pattern} & \multirow{3}{*}{Max} &      & $\bar \tau$ = 0.5 &      &      & $\bar \tau$= 0.7 &  \\ \cline{4-9}
			&                                       &                      &          \multicolumn{6}{c|}{Average Proximal Effect }           \\ \cline{4-9}
			&                                       &                      & 0.10 &       0.08        & 0.06 & 0.10 &       0.08       & 0.06 \\ \hline
			\multirow{12}{*}{i.i.d. Normal}    &      \multirow{3}{*}{Pattern 1}       &          22          & 80.9 &       80.0        & 81.0 & 78.7 &       77.5       & 80.7 \\
			&                                       &          29          & 78.4 &       80.6        & 77.8 & 80.6 &       78.7       & 79.0 \\
			&                                       &          36          & 80.2 &       80.0        & 79.6 & 79.4 &       80.2       &\textbf{ 77.0} \\ \cline{2-9}
			&      \multirow{3}{*}{Pattern 2}       &          22          & 80.3 &       78.1        & 78.8 & 80.6 &       79.6       & 79.8 \\
			&                                       &          29          & 80.3 &       79.1        & 80.2 &\textbf{ 77.4} &       79.9       & 79.9 \\
			&                                       &          36          & \textbf{76.8} &       79.3        & 80.2 & 78.5 &       78.4       & 80.0 \\ \cline{2-9}
			&      \multirow{3}{*}{Pattern 3}       &          22          & 83.5 &       81.5        & 77.7 & 78.5 &       81.3       & 78.7 \\
			&                                       &          29          & 77.9 &       79.1        & 78.5 & 77.8 &       78.8       & 79.0 \\
			&                                       &          36          &\textbf{ 77.3} &       78.1        & 79.8 & 79.8 &       79.9       & 79.1 \\ \cline{2-9}
			&      \multirow{3}{*}{Pattern 4}       &          22          & \textbf{77.2} &       79.7        & 81.8 & 80.2 &       79.0       & 78.8 \\
			&                                       &          29          & 80.1 &       78.8        & 80.3 & 79.4 &       80.6       & 80.1 \\
			&                                       &          36          & 80.5 &       79.4        & 80.0 & 78.9 &       79.9       & 78.1 \\ \hline
			\multirow{3}{*}{i.i.d. t dist.}    &      \multirow{3}{*}{Pattern 1}       &          22          & 80.4 &       81.9        & 81.0 & 79.7 &       79.4       & 80.7 \\
			&                                       &          29          & 81.7 &       82.2        & 82.2 & 79.1 &       82.3       &\textbf{ 77.3} \\
			&                                       &          36          & 80.8 &       78.8        & 79.5 & 81.8 &       81.6       & 79.9 \\ \hline
			\multirow{3}{*}{i.i.d. Exp.}      &      \multirow{3}{*}{Pattern 1}       &          22          & 81.0 &       81.6        & 79.7 & \textbf{77.2} &       80.1       & 80.2 \\
			&                                       &          29          & 80.6 &       82.4        & 80.3 & 79.0 &       79.8       & 80.3 \\
			&                                       &          36          & 82.1 &       79.8        & 80.8 & 79.8 &       79.5       & 80.3 \\ \hline
			\multirow{3}{*}{AR(1), $\phi = -0.6$} &      \multirow{3}{*}{Pattern 1}       &          22          & 78.5 &       80.3        & 78.5 & 82.3 &       79.8       & 80.3 \\
			&                                       &          29          & 78.7 &       80.8        & 80.0 & \textbf{77.1} &       79.5       & 77.9 \\
			&                                       &          36          & 77.7 &       80.3        & 80.2 & 78.2 &       \textbf{77.4}       & 83.6 \\ \hline
			\multirow{3}{*}{AR(1), $\phi = -0.3$} &      \multirow{3}{*}{Pattern 1}       &          22          & 77.9 &       79.0        & 79.6 & 80.0 &       77.8       & 80.4 \\
			&                                       &          29          & 77.9 &       79.1        & 80.0 & 79.0 &       78.0       & 78.4 \\
			&                                       &          36          & 78.1 &       81.2        & 80.2 & 80.7 &       80.9       & 78.4 \\ \hline
			\multirow{3}{*}{AR(1), $\phi = 0.3$}  &      \multirow{3}{*}{Pattern 1}       &          22          & 80.2 &       78.5        & 80.8 & 80.5 &       79.6       & 82.6 \\
			&                                       &          29          & 78.0 &       80.0        & 80.0 & 78.0 &       79.4       & 80.1 \\
			&                                       &          36          & 77.6 &       82.5        & 80.6 & 77.0 &       78.9       & 82.0 \\ \hline
			\multirow{3}{*}{AR(1), $\phi = 0.6$}  &      \multirow{3}{*}{Pattern 1}       &          22          & 80.4 &       79.8        & 79.5 & 80.7 &       79.5       & 82.0 \\
			&                                       &          29          & 78.9 &       81.5        & 79.3 & 79.5 &       81.3       & 79.5 \\
			&                                       &          36          & 79.5 &       78.4        & 78.8 & 80.1 &       77.9       & 77.8 \\ \hline
			\multirow{3}{*}{AR(5), $\phi = -0.6$} &      \multirow{3}{*}{Pattern 1}       &          22          & 79.9 &       79.4        & 80.0 & 78.7 &       79.2       & 79.4 \\
			&                                       &          29          & 80.0 &       78.3        & 79.1 & \textbf{76.8} &       79.6       & 79.3 \\
			&                                       &          36          & 80.5 &       80.0        & 79.2 & 80.1 &       78.0       & 80.4 \\ \hline
			\multirow{3}{*}{AR(5), $\phi = -0.3$} &      \multirow{3}{*}{Pattern 1}       &          22          & 79.2 &       80.4        & 81.9 & 81.3 &       77.7       & 79.1 \\
			&                                       &          29          & 80.0 &       82.3        & 80.5 & 80.5 &       82.2       & 79.2 \\
			&                                       &          36          & \textbf{75.9} &       78.7        & 79.3 & 79.0 &       79.4       & 79.9 \\ \hline
			\multirow{3}{*}{AR(5), $\phi = 0.3$}  &      \multirow{3}{*}{Pattern 1}       &          22          & 79.4 &       80.8        & 79.8 & 79.5 &       \textbf{77.3 }      & 81.2 \\
			&                                       &          29          & 78.0 &       79.2        & 79.2 & 79.2 &       80.5       & 78.4 \\
			&                                       &          36          & 78.3 &       79.1        & 78.1 & 80.7 &       80.5       & 79.5 \\ \hline
			\multirow{3}{*}{AR(5), $\phi = 0.6$}  &      \multirow{3}{*}{Pattern 1}       &          22          & 80.2 &       77.9        & 80.3 & 78.6 &       78.4       & 80.3 \\
			&                                       &          29          & \textbf{76.9} &       79.3        & 80.2 & 79.1 &       80.6       & 80.5 \\
			&                                       &          36          & 78.7 &       84.0        & 80.1 & 78.8 &       79.3       & 78.8 \\ \hline
		\end{tabular}
		\begin{tablenotes}
			\item   \lq\lq Max\rq\rq is the day in which the maximal proximal effect is attained.   $\bar\tau = (1/T)\sum_{t=1}^{T}E[I_t]$ is the average availability. $\phi$ is the parameter for AR(1) and AR(5) process. Bold numbers are significantly(at .05 level) less than .80.
		\end{tablenotes}
		\label{worktrue42P}	
	\end{threeparttable}
\end{table}

\begin{table}[H]
	\centering
	\begin{threeparttable}

		\caption{Simulated type 1 error rate($\%$) and power($\%$) when the duration of study is 4-week and 8-week.  Error terms follow i.i.d. N(0,1). The associated sample size is given in Table \ref{worktrueformula}. }
		\begin{tabular}{|c|c|c|ccc|ccc|}

			\multirow{3}{*}{Duration of Study} & \multirow{3}{*}{Availability Pattern} & \multirow{3}{*}{Max} &      & $\bar \tau$ = 0.5 &      &      & $\bar \tau$= 0.7 &  \\ \cline{4-9}
			&                                       &                      &          \multicolumn{6}{c|}{Average Proximal Effect }           \\
			\cline{4-9}
			&                                       &                      & 0.10 &       0.08        & 0.06 & 0.10 &       0.08       & 0.06 \\
			\hline
			
			\multirow{12}{*}{4-week}       &      \multirow{3}{*}{Pattern 1}       &          15          & 4.1 & 4.7 & 6.3 & 5.3 & 5.5 & 5.6 \\
			&                                       &          22          & 5.2 & 4.4 & 4.7 & 3.1 & 4.7 & 4.4 \\
			&                                       &          29           & 5.7 & 5.5 & 5.6 & 4.3 & 4.2 & 4.2 \\
			\cline{2-9}
			&      \multirow{3}{*}{Pattern 2}       &          15          & 4.8 & 4.8 & 5.0 & 5.0 & 5.2 & 5.3 \\
			&                                       &          22           & 5.1 & 5.2 & 4.7 & 3.7 & 4.2 & 3.7 \\
			&                                       &          29          & 5.6 & 5.1 & 4.2 & 4.2 & 4.9 & 4.4 \\
			\cline{2-9}
			&      \multirow{3}{*}{Pattern 3}       &          15          & 4.7 & 5.0 & 4.6 & 6.1 & 5.3 & 5.1 \\
			&                                       &          22          & 4.9 & 4.0 & \textbf{6.6} & 4.2 & 3.8 & 4.1 \\
			&                                       &          29           & 4.7 & 4.3 & 5.1 & 4.6 & 5.8 & 3.5 \\
			\cline{2-9}
			&      \multirow{3}{*}{Pattern 4}       &          15         & 4.9 & 4.6 & 4.8 & 3.0 & 5.9 & 3.8 \\
			&                                       &          22          & 3.5 & 5.1 & 4.5 & 5.2 & 3.8 & 6.0 \\
			&                                       &          29         & 4.4 & 6.4 & 4.7 & 4.4 & 4.3 & 4.7 \\ 			
			\hline
			
			\multirow{12}{*}{8-week}       &      \multirow{3}{*}{Pattern 1}       &          29          & 4.1 & 4.6 & 4.0 & 5.3 & 5.0 & 5.9 \\
			&                                       &          36          & 3.3 & 4.7 & \textbf{6.5} & 4.6 & 5.4 & 4.3 \\
			&                                       &          43          & 3.2 & 5.1 & 5.2 & 5.0 & 3.4 & 5.0 \\
			\cline{2-9}
			&      \multirow{3}{*}{Pattern 2}       &          29          & 3.9 & 5.0 & 4.5 & 4.2 & 3.7 & 4.1 \\
			&                                       &          36          & 3.8 & 4.6 & 4.9 & 4.5 & 3.4 & 5.2 \\
			&                                       &          43          & 3.9 & 5.4 & 5.0 & 3.4 & 3.8 & 5.0 \\
			\cline{2-9}
			&      \multirow{3}{*}{Pattern 3}       &          29          & 4.6 & 4.2 & 3.7 & 5.2 & 4.1 & 4.0 \\
			&                                       &          36          & 4.3 & 5.1 & 6.1 & 4.6 & 5.0 & 4.6 \\
			&                                       &          43          & 4.6 & 6.0 & 4.1 & 5.0 & 4.9 & 4.0 \\
			\cline{2-9}
			&      \multirow{3}{*}{Pattern 4}       &          29          & 4.5 & 5.2 & 2.9 & 3.6 & 5.3 & 4.4 \\
			&                                       &          36          & 4.5 & 5.2 & 3.7 & 2.7 & 3.7 & 4.7 \\
			&                                       &          43         & 4.2 & \textbf{7.1 }& 4.9 & 4.4 & 4.5 & 4.8 \\ 			
			\hline
			
			\multirow{12}{*}{4 week}       &      \multirow{3}{*}{Pattern 1}       &          15          & 80.4 & 79.0 & 78.5 & 79.6 & 82.8 & 80.3 \\
			&                                       &          22          & 78.8 & 78.7 & 80.7 & 78.7 & 79.2 & 80.0 \\
			&                                       &          29          & \textbf{76.2} & 80.6 & 80.1 & 81.3 & 80.1 & 79.1 \\
			\cline{2-9}
			&      \multirow{3}{*}{Pattern 2}       &          15          & 82.4 & 77.8 & \textbf{77.2} & \textbf{75.9 }& 80.0 & 78.9 \\
			&                                       &          22          & \textbf{77.2} & 80.3 & 81.5 & \textbf{75.8} & 80.7 & 82.0 \\
			&                                       &          29          & 80.1 & 79.3 & 80.1 & 78.0 & 77.7 & \textbf{76.9 }\\
			\cline{2-9}
			&      \multirow{3}{*}{Pattern 3}       &          15          & 79.3 & 79.8 & 79.2 & 79.1 & \textbf{76.5} & 80.8 \\
			&                                       &          22          & 80.0 & 80.0 & 79.0 & 79.0 & 80.2 & 81.8 \\
			&                                       &          29          & 79.4 & 80.7 & 79.3 & 80.4 & 79.6 & 79.2 \\
			\cline{2-9}
			&      \multirow{3}{*}{Pattern 4}       &          15          & 82.6 & 78.3 & 79.2 & 80.5 & 80.0 & 79.5 \\
			&                                       &          22          & 80.4 & 80.7 & 79.3 & 79.1 & 78.5 & 79.2 \\
			&                                       &          29          & 78.4 & 79.2 & 78.5 & 79.6 & 79.2 & 80.5 \\			
			\hline

			\multirow{12}{*}{8 week}       &      \multirow{3}{*}{Pattern 1}       &          29          & 79.7 &\textbf{ 77.3} & \textbf{76.4} & 79.1 & 82.2 & 79.6 \\
			&                                       &          36          & 78.8 & 78.6 & 81.5 & 80.3 & 78.2 & 79.6 \\
			&                                       &          43          & 80.4 & 77.8 & 78.7 & 79.1 & 80.3 & 80.1 \\
			\cline{2-9}
			&      \multirow{3}{*}{Pattern 2}       &          29          & 79.3 & 81.1 & 79.8 & 78.7 & 79.7 & 80.2 \\
			&                                       &          36          & 81.2 & 78.5 & 79.0 & 81.3 & 80.8 & 78.2 \\
			&                                       &          43          & 80.3 & 81.5 & 77.5 & \textbf{75.1} & 78.8 & 78.1 \\
			\cline{2-9}
			&      \multirow{3}{*}{Pattern 3}       &          29          & 80.1 & 79.0 &\textbf{ 77.1} & 78.2 & 80.4 & 78.8 \\
			&                                       &          36          & 79.5 & 79.9 & 79.6 & 80.0 & 80.8 & 79.6 \\
			&                                       &          43          & 80.5 & 79.5 & 79.6 & 79.4 & 79.4 & 80.2 \\
			\cline{2-9}
			&      \multirow{3}{*}{Pattern 4}       &          29          & 82.1 & 79.7 & 80.7 & 79.7 & 79.0 & 78.4 \\
			&                                       &          36          & 77.8 & 78.2 & 80.1 & 77.9 & \textbf{76.9} & 79.5 \\
			&                                       &          43          & 79.6 & 78.5 & 78.1 & 79.4 & 80.6 & 79.5 \\
			\hline			
		\end{tabular}
		\begin{tablenotes}

			\item   \lq\lq Max\rq\rq is the day in which the maximal proximal effect is attained.   $\bar\tau = (1/T)\sum_{t=1}^{T}E[I_t]$ is the average availability.  Bold numbers are significantly(at .05 level) greater than .05 (for type I error)and less than 0.80 (for power).
		\end{tablenotes}
		\label{worktrue2856}		
	\end{threeparttable}

\end{table}

\begin{table}[H]
	\centering
	\begin{threeparttable}
		\caption{Simulated Type I error rate($\%$) and power($\%$) when the availability indicator,  $I_t$ depends on the recent past treatments with $\eta = -0.2$. The expected availability is constant in $t$ and equal to $0.5$. Duration of study is 42 days.  The associated sample size is given in Table \ref{worktrueformula}. }
		\begin{tabular}{|c|c|c|ccc|ccc|ccc|ccc|}

			\multirow{3}{3em}{Error  Term} &  \multirow{3}{*}{$\phi$}   & \multirow{3}{*}{Max} &      & $\bar \tau$ = 0.5 &      &      & $\bar \tau$= 0.7 &  & &$\bar \tau$ = 0.5   & & &$\bar \tau$ = 0.7  & \\ \cline{4-15}
			&                            &                      &          \multicolumn{12}{c|}{Average Proximal Effect }            \\ \cline{4-15}
			&                            &                      & 0.10 &       0.08        & 0.06 & 0.10 &       0.08       & 0.06  & 0.10 &       0.08        & 0.06 & 0.10 &       0.08       & 0.06  \\
			\hline
			\multirow{12}{*}{AR(1)}    & \multirow{3}{*}{-0.6} &          22          & 4.8 & 5.4 & 4.5 & 3.4 & 5.8 & 3.7 & 81.5 & 78.0 & 79.4 & 81.7 & 77.9 & 80.7 \\
			&                            &          29          & 4.7 & 4.4 & 4.2 & 4.0 & 4.9 & 4.6 & 79.4 & 80.9 & 80.7 & 78.2 & 79.2 & 79.7 \\
			&                            &          36          & 4.3 & 5.3 & 4.4 & 4.2 & 3.9 & 5.5 & 79.5 & 81.5 & 79.8 & 80.2 & 79.2 & 80.7 \\
			\cline{2-15}
			& \multirow{3}{*}{-0.3} &          22          & 4.7 & 3.8 & 4.4 & 3.5 & 4.4 & 4.6 & 78.7 & 81.2 & 80.3 & 80.9 & 77.9 & 78.5 \\
			&                            &          29           & 3.8 & 4.0 & 4.9 & 3.5 & 5.0 & 4.4 & 80.1 & 79.5 & 81.2 & \textbf{77.3} & 79.5 & \textbf{77.1} \\
			&                            &          36          & 2.7 & 5.7 & 4.0 & 3.3 & 4.7 & 5.2 & \textbf{76.8} & 80.4 & 79.9 & 78.8 & 79.5 & 79.4 \\
			\cline{2-15}
			& \multirow{3}{*}{0.3} &          22          & 4.8 & 4.1 & 4.4 & 5.0 & 5.4 & 3.6 & 83.0 & 79.8 & 79.4 & 81.3 & 78.9 & 79.2 \\
			&                            &          29          & 4.9 & 4.6 & 5.0 & 4.4 & 5.5 & 5.6 & 79.5 & 80.3 & 82.2 & 78.5 & 80.7 & 77.6 \\
			&                            &          36          & 4.9 & 4.9 & 4.2 & 3.3 & 4.5 & 4.8 & 80.0 & 78.9 & 79.5 & 81.7 & 79.4 & 79.6 \\
			\cline{2-15}
			& \multirow{3}{*}{0.6} &          22          & 4.5 & 5.1 & 4.7 & 4.3 & 4.6 & 4.0 & 80.3 & 78.9 & 81.1 & 81.2 & 81.5 & 77.9 \\
			&                            &          29          & 3.4 & 4.5 & 5.1 & 4.4 & 4.3 & 4.6 & 79.3 & \textbf{76.2} & 79.4 & 81.3 & 80.6 & 79.4 \\
			&                            &          36          & 4.8 & 4.3 & 4.2 & 4.1 & 4.5 & 4.5 & 77.5 & 80.5 & 80.9 & 76.7 & 80.0 & 79.7 \\
			\hline 
			
			\multirow{12}{*}{AR(5)}    & \multirow{3}{*}{-0.6} &          22          & 4.8 & 4.6 & 4.3 & 3.7 & 4.7 & 3.5 & 81.9 & 81.4 & 81.6 & 79.8 & 78.3 & 78.9 \\
			&                            &          29          & 6.5 & 4.1 & 4.5 & 3.3 & 4.5 & 4.8 & 77.5 & 79.9 & 79.8 & 79.9 & 79.3 & 79.3 \\
			&                            &          36         & 3.5 & 5.7 & 4.4 & 4.6 & 4.7 & 5.7 & 77.8 & 80.8 & 78.6 & 77.9 & 79.2 & 81.7 \\
			\cline{2-15}
			& \multirow{3}{*}{-0.3} &          22          & 4.3 & 4.9 & 4.0 & 4.3 & 5.6 & 5.0 & 77.7 & 81.8 & 80.0 & 80.1 & 80.3 & 81.1 \\
			&                            &          29          & 3.9 & 4.0 & 5.0 & 3.2 & 5.7 & 5.1 & 80.0 & 80.9 & 80.3 & 80.6 & 80.3 & 77.8 \\
			&                            &          36          & 4.0 & 3.6 & 4.7 & 4.8 & 4.8 & 3.2 & 79.0 & 80.4 & 80.8 & 80.1 & 79.0 & \textbf{76.5} \\
			\cline{2-15}
			& \multirow{3}{*}{0.3} &          22          & 3.5 & 4.9 & 5.0 & 4.1 & 3.8 & 4.1 & \textbf{77.4} & 82.9 & 78.5 & 80.6 & 81.4 & 80.2 \\
			&                            &          29          & 4.6 & 6.1 & 4.7 & 4.7 & 4.1 & 4.1 & 78.7 & 82.0 & 78.0 & 81.4 & \textbf{76.5} & 81.3 \\
			&                            &          36          & 5.1 & 4.4 & 4.0 & 3.2 & 3.9 & 4.7 & 79.7 & 81.8 & 78.6 & 79.1 & \textbf{77.4} & 79.0 \\
			\cline{2-15}
			& \multirow{3}{*}{0.6} &          22         & 5.0 & 4.6 & 4.3 & 4.0 & 4.0 & 5.5 & 80.5 & 79.4 & 82.5 & 79.2 & 81.1 & 81.0 \\
			&                            &          29          & 5.6 & 4.3 & 6.9 & 5.6 & 3.4 & 3.1 & 78.3 & 80.0 & 80.5 & 80.8 & 80.4 & 78.4 \\
			&                            &          36          & 4.8 & 4.8 & 4.8 & 3.5 & 3.7 & 5.5 & 78.2 & 80.5 & 80.3 & 77.6 & 80.5 & 79.1 \\ 			
			\hline			
			
		\end{tabular}
		\begin{tablenotes}
			\item   \lq\lq Max\rq\rq is the day in which the maximal proximal effect is attained.   $\bar\tau = (1/T)\sum_{t=1}^{T}E[I_t]$ is the average availability. $\phi$ is the parameter for AR(1) and AR(5) process.  Bold numbers are significantly(at .05 level) greater than .05 and less than 0.80.
		\end{tablenotes}
		\label{worktruedep}		
	\end{threeparttable}
	
\end{table}

\newpage

\begin{table}[H]
	\centering	
	\begin{threeparttable}
		\caption{Simulated type I error rate($\%$) and power($\%$) when working assumption (a) is violated. Scenario 1. The average availability is 0.5. The day of maximal proximal effect is 29. }
		
		\begin{tabular}{|c|c|cccc|cccc|}

			\multirow{2}{*}{$\theta$}&\multirow{2}{*}{$\bar d$} &\multicolumn{8}{c|}{Availability Pattern}\\
			\cline{3-10}	
			& & Pattern 1 & Pattern 2 & Pattern 3 & Pattern 4 & Pattern 1 & Pattern 2 & Pattern 3 & Pattern 4 \\
			\hline
			
			\multirow{3}{*}{$0.5 \bar d$} & 0.10 & 5.5 & 4.6 & 4.2 & 5.1 & 79.7 & 79.4 & 80.5 & 80.1 \\
			& 0.08 & 5.1 & 4.4 & 5.4 & 4.6 & 80.4 & 78.9 & 80.4 & 78.7 \\
			& 0.06 & 4.1 & 5.5 & 4.6 & 4.3 & 77.5 & 82.7 & 81.0 & 81.0 \\
			\hline
			\multirow{3}{*}{$ \bar d$}& 0.10 & 4.8 & 4.3 & 3.7 & 4.1 & 79.3 & 78.3 & 77.8 & 79.4 \\
			& 0.08 & 5.4 & 4.9 & 4.6 & 5.5 & 78.8 & 79.3 & 78.0 & 80.6 \\
			& 0.06 & 4.4 & 3.5 & 5.1 & 4.6 & 78.4 & 79.3 & 79.0 & 80.4 \\
			\hline
			\multirow{3}{*}{$1.5 \bar d$} & 0.10 & 4.4 & 4.1 & 4.4 & 4.8 & 78.3 & 80.5 & 78.4 & 79.9 \\
			& 0.08 & 5.0 & 4.3 & 4.3 & 3.9 & 80.5 & 79.7 & 78.7 & 81.9 \\
			& 0.06 & 4.0 & 5.1 & 5.5 & 5.6 & \textbf{77.2} & 80.8 & 81.6 & 80.3 \\
			\hline
			\multirow{3}{*}{$2 \bar d$} & 0.10  & 4.1 & 3.8 & 5.0 & 5.5 & 77.7 & 78.8 & 79.0 & 78.4 \\
			& 0.08 & 4.0 & 5.0 & 3.7 & 5.7 & 79.3 & 81.5 & 79.1 & 79.4 \\
			& 0.06 & 4.9 & 4.3 & 5.2 & 5.3 & 80.8 & 79.0 & 77.5 & 80.9 \\ 		
			\hline
		\end{tabular}
		\begin{tablenotes}
			\item   $\bar d = (1/T)\sum_{t=1}^{T}Z_t' d$ is the average proximal effect. $\theta$ is the coefficient of $W_t$ in $E[Y_{t+1}|I_t = 1]$. Bold numbers are significantly (at .05 level) greater than .05 (for type I error rate) and lower than 0.80(for power).
		\end{tablenotes}
		
		\label{workfalsea1}
	\end{threeparttable}
\end{table}

\begin{figure}[H]
	\centering
	\includegraphics[width = 0.9\linewidth]{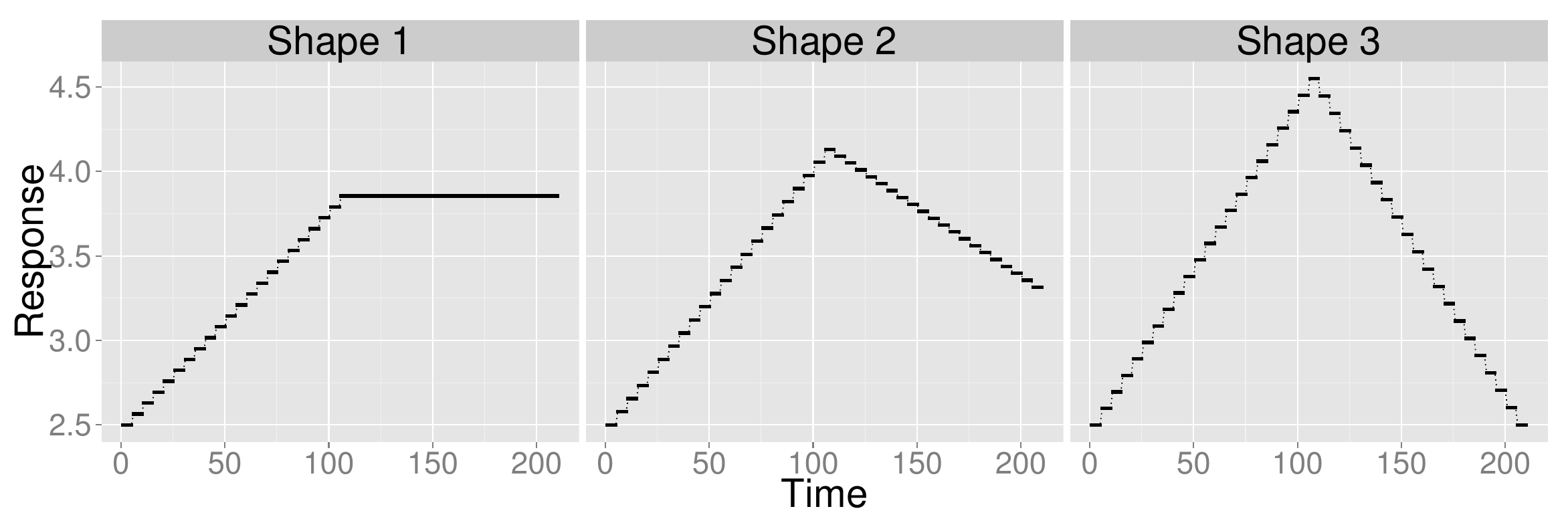}
	\caption{Conditional expectation of proximal response, $E[Y_{t+1}|I_t = 1]$. The horizontal axis is the decision time point. The vertical axis is $E[Y_{t+1}|I_t = 1]$.}		
	\label{ShapeOfAlpha}
\end{figure}

\begin{table}[H]
	\centering
	\begin{threeparttable}
		
		\caption{Simulated Type I error rate($\%$) and power ($\%$) when working assumption (a) is violated. Scenario 2.  The shapes of $\alpha(t)=E[Y_{t+1}|I_t=1]$ and patterns of availability are provided in Figure \ref{ShapeOfAlpha} and Figure \ref{ShapeOfTau}. The average availability is 0.5. The day of maximal proximal effect is 29.  The associated sample size is given in Table \ref{worktrueformula}. }	
		
		\begin{tabular}{|c|c|cccc|cccc|}

			&& 	\multicolumn{8}{c|}{Availability Pattern } \\
			\cline{3-10}
			$\alpha(t)$ &$\bar d$ & Pattern 1 & Pattern 2 & Pattern 3 & Pattern 4 & Pattern 1 & Pattern 2 & Pattern 3 & Pattern 4 \\

			\hline
			\multirow{3}{*}{Shape 1}& 0.10 & 3.6 & 4.3 & 4.7 & 4.5 & \textbf{77.4} & 80.2 &\textbf{ 76.2} & \textbf{75.9} \\ 
			& 0.08 & 5.9 & 3.8 & 4.1 & 3.4 & 79.7 & 80.1 & 78.9 & 80.6 \\ 
			& 0.06 & 4.6 & 5.7 & 4.2 & \textbf{6.5} & 78.7 & \textbf{76.3} & 78.3 & 79.9 \\ 
			\hline
			\multirow{3}{*}{Shape 2}& 0.10 & 4.8 & 4.8 & 4.4 & 4.1 & 79.2 & 79.1 & 78.5 & 79.7 \\ 
			& 0.08 & 3.9 & 5.4 & 4.8 & 4.3 & 77.7 & 80.4 & \textbf{76.8} & 80.9 \\ 
			& 0.06 & 5.1 & 5.5 & 3.4 & 4.9 & 78.3 & 79.4 & 79.8 & 80.2 \\ 
			\hline
			\multirow{3}{*}{Shape 3}& 0.10 & 5.1 & 3.5 & 4.3 & 4.4 & 79.1 & 79.4 & \textbf{75.6} & 78.0 \\ 
			& 0.08 & 4.6 & 5.0 & 6.2 & 3.8 & 78.3 & 78.1 & 79.1 & 78.1 \\ 
			& 0.06 & 4.8 & 4.4 & 5.4 & 4.2 & 78.0 & 78.3 & 79.8 & 77.7 \\ 			
			\hline
		\end{tabular}
		\begin{tablenotes}
			\item   $\bar d=(1/T)\sum_{t=1}^{T}Z_t'd$ is the average standardized treatment effect. Bold numbers are significantly (at .05 level) greater than .05 (for type I error rate) and lower than 0.80(for power).
		\end{tablenotes}	
		\label{workfalsea2}
	\end{threeparttable}
\end{table}

\begin{figure}[H]
	\centering
	\includegraphics[width = 1\linewidth]{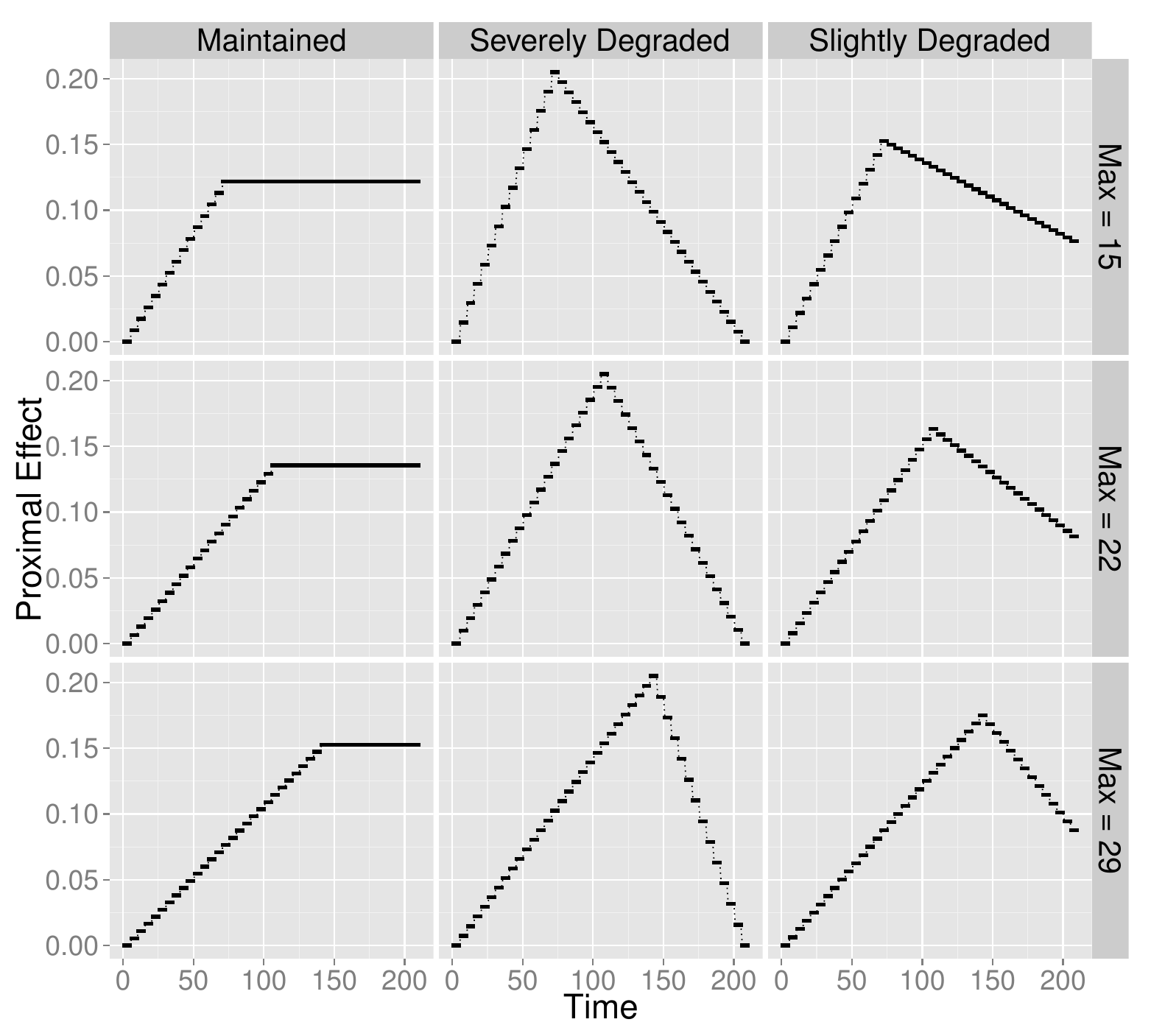}
	\caption{Proximal Main Effects of Treatment, $\{d(t)\}_{t=1}^T$: representing maintained, slightly degraded and severely degraded time-varying treatment effects. The horizontal axis is the decision time point. The vertical axis is the standardized treatment effect. The "Max" in the title refers to the day of maximal effect. The average standardized proximal effect is $0.1$ in all plots.}	
	\label{ShapeOfBetafull}
\end{figure}

\begin{table}[H]
	\centering
	\begin{threeparttable}	
		\caption{Sample Sizes when working assumption (b) is violated. The vector of standardized effects sizes, $d$, used in the sample size formula provides the projection of $d(t)$. The sample size formula is used with the correct availability pattern, $\{E[I_t]\}_{t=1}^T$. The shape of the standardized proximal effect, $d(t)=\beta(t)/\bar\sigma$ and pattern for availability, $E[I_t]$ are provided in Figure~\ref{ShapeOfBetafull} and in Figure~(\ref{ShapeOfTau}). The significance level is 0.05. The desired power is 0.80.  }
		
		\begin{tabular}{|c|c|c|ccc|ccc|}

			& & & \multicolumn{3}{c|}{$\bar \tau$ = 0.5} & \multicolumn{3}{c|}{$\bar \tau$ = 0.7} \\
			\cline{4-9}		
			& Availability & & \multicolumn{6}{c|}{Shape of $d(t)$}\\
			\cline{4-9}
			$\bar d$& Pattern & Max  & \multirow{2}{*}{Maintained} & \multirow{2}{4.5em}{\ \ Slightly Degraded}  & \multirow{2}{4.5em}{\ Severely Degraded} & \multirow{2}{*}{Maintained} & \multirow{2}{4.5em}{\ \ Slightly Degraded}  & \multirow{2}{4.5em}{\ Severely Degraded}\\
			& & & & & & & & \\	
			\hline
			\multirow{12}{*}{0.10} & & 15 &  43 &  41 &  39 &  32 &  31 &  29 \\
			& Pattern 1 & 22 &  43 &  41 &  40 &  33 &  31 &  30 \\
			& & 29 &  38 &  37 &  38 &  29 &  28 &  29 \\
			\cline{2-9}
			& & 15 &  43 &  41 &  39 &  33 &  31 &  30 \\
			& Pattern 2 & 22 &  43 &  42 &  40 &  33 &  31 &  30 \\
			& & 29 &  38 &  37 &  38 &  29 &  28 &  29 \\
			\cline{2-9}
			& & 15 &  45 &  43 &  41 &  33 &  32 &  31 \\
			& Pattern 3 & 22 &  44 &  43 &  42 &  33 &  32 &  31 \\
			& & 29 &  37 &  38 &  39 &  28 &  28 &  29 \\
			\cline{2-9}
			& & 15 &  42 &  39 &  37 &  32 &  30 &  28 \\
			& Pattern 4 & 22 &  44 &  41 &  39 &  33 &  31 &  30 \\
			& & 29  &  39 &  38 &  38 &  29 &  28 &  28 \\
			\hline

			\multirow{12}{*}{0.08 } & & 15 &  65 &  61 &  58 &  48 &  45 &  43 \\
			& Pattern 1 & 22 &  65 &  62 &  60 &  48 &  46 &  44 \\
			& & 29 &  56 &  55 &  56 &  42 &  41 &  42 \\
			\cline{2-9}
			& & 15 &  65 &  61 &  59 &  48 &  45 &  43 \\
			& Pattern 2 & 22 &  65 &  62 &  60 &  48 &  46 &  44 \\
			& & 29 &  56 &  55 &  56 &  42 &  41 &  42 \\
			\cline{2-9}
			
			& & 15 &  67 &  64 &  62 &  49 &  47 &  45 \\
			& Pattern 3 & 22 &  66 &  64 &  63 &  48 &  47 &  46 \\
			& & 29 &  56 &  56 &  59 &  41 &  41 &  43 \\
			\cline{2-9}
			
			& & 15 &  63 &  59 &  55 &  47 &  44 &  41 \\
			& Pattern 4 & 22 &  65 &  61 &  58 &  48 &  45 &  43 \\
			& & 29  &  58 &  56 &  56 &  43 &  41 &  41 \\
			\hline
			
			\multirow{12}{*}{0.06 } & & 15 & 111 & 105 & 100 &  81 &  76 &  73 \\
			& Pattern 1 & 22 & 112 & 106 & 103 &  81 &  77 &  75 \\
			& & 29 &  96 &  94 &  96 &  70 &  69 &  70 \\
			\cline{2-9}
			
			& & 15  & 112 & 105 & 100 &  81 &  77 &  73 \\
			& Pattern 2 & 22 & 112 & 106 & 103 &  81 &  77 &  75 \\
			& & 29 &  96 &  94 &  96 &  70 &  68 &  70 \\
			\cline{2-9}
			
			& & 15  & 116 & 111 & 106 &  83 &  79 &  76 \\
			& Pattern 3 & 22 & 114 & 110 & 108 &  82 &  79 &  78 \\
			& & 29 &  95 &  96 & 101 &  69 &  69 &  72 \\
			\cline{2-9}
			
			& & 15 & 108 & 100 &  94 &  79 &  74 &  70 \\
			& Pattern 4 & 22 & 112 & 105 &  99 &  81 &  76 &  73 \\
			& & 29 & 100 &  95 &  95 &  72 &  69 &  70 \\		
			\hline
		\end{tabular}
		\begin{tablenotes}
			\item   \lq\lq Max\rq\rq is the day in which the maximal proximal effect is attained. $\bar d=(1/T)\sum_{t=1}^{T}Z_t'd$ is the average standardized treatment effect.
		\end{tablenotes}
		\label{workfalsecase1SampleSize}	
	\end{threeparttable}
	
\end{table}

\begin{table}[H]
	\centering
	\begin{threeparttable}
		
		\caption{Simulated power($\%$) when working assumption (b) is violated. The shape of the standardized proximal effect, $d(t)=\beta(t)/\bar\sigma$ and pattern for availability, $E[I_t]$ are provided in Figure~\ref{ShapeOfBetafull} and in Figure~(\ref{ShapeOfTau}). The corresponding sample sizes are given in Table \ref{workfalsecase1SampleSize}.}
		
		\begin{tabular}{|c|c|c|ccc|ccc|}

			& & & \multicolumn{3}{c|}{$\bar \tau$ = 0.5} & \multicolumn{3}{c|}{$\bar \tau$ = 0.7} \\
			\cline{4-9}		
			& Availability & & \multicolumn{6}{c|}{Shape of $d(t)$}\\
			\cline{4-9}
			$\bar d$& Pattern & Max  & \multirow{2}{*}{Maintained} & \multirow{2}{4.5em}{\ \ Slightly Degraded}  & \multirow{2}{4.5em}{\ Severely Degraded} & \multirow{2}{*}{Maintained} & \multirow{2}{4.5em}{\ \ Slightly Degraded}  & \multirow{2}{4.5em}{\ Severely Degraded}\\
			& & & & & & & & \\		
			\hline
			
			\multirow{12}{*}{0.10} & & 15 & 78.4 & 78.8 & 78.6 & 79.1 & 80.1 & 77.6 \\
			& Pattern 1 & 22 & 80.4 & 79.5 & 81.2 & 80.0 & \textbf{76.9} & 77.9 \\
			& & 29 & 80.4 & 79.2 & 78.9 & \textbf{77.3} & \textbf{76.8} & 81.1 \\
			\cline{2-9}
			& & 15 & 78.6 & 79.9 & 79.9 & 80.1 & 80.4 & 81.3 \\
			& Pattern 2 & 22 & 78.3 & 81.2 & 78.8 & 79.2 & 80.8 & 80.5 \\
			& & 29 & 77.9 & 80.8 & 79.3 & 78.1 & 77.7 & 82.2 \\
			\cline{2-9}
			& & 15 & 81.0 & 79.7 & \textbf{77.4} & 77.9 & 80.9 & 77.6 \\
			& Pattern 3 & 22 & 78.9 & 79.1 & 80.0 & 79.7 & 79.4 & \textbf{75.9} \\
			& & 29 & 80.9 & 77.5 & 77.7 & 80.6 & 79.2 & 78.5 \\
			\cline{2-9}
			& & 15 & 79.7 & 79.5 & 77.9 & 79.5 & 81.7 & 78.0 \\
			& Pattern 4 & 22 & 78.9 & 77.9 & 80.4 & 82.2 & 78.9 & 78.8 \\
			& & 29 & 77.9 & 79.7 & 79.0 & 78.0 & 80.2 & 80.8 \\

			\hline
			\multirow{12}{*}{0.08 } & & 15 & 80.5 & 79.5 & 78.6 & 80.6 & 79.2 & 78.7 \\
			& Pattern 1 & 22 & 78.9 & 78.7 & 78.8 & 78.9 & 80.7 & 80.3 \\
			& & 29 & \textbf{76.6} & 78.0 & 78.3 & 80.9 & 78.6 & 80.4 \\
			\cline{2-9}
			& & 15 & 81.0 & 79.3 & 78.7 & 82.0 & 80.5 & 80.1 \\
			& Pattern 2 & 22 & 82.4 & 80.6 & 80.0 & 78.0 & 79.6 & 79.4 \\
			& & 29 & 79.2 & \textbf{76.9} & 81.9 & 78.3 & 78.8 & 79.7 \\
			\cline{2-9}
			
			& & 15 & 78.2 & 81.6 & 80.9 & 79.1 & 79.2 & 77.5 \\
			& Pattern 3 & 22 & 80.9 & 79.5 & 78.6 & 79.2 & 78.3 & 81.4 \\
			& & 29  & 80.4 & 79.3 & 77.5 & 77.9 & 80.2 & 82.3 \\
			\cline{2-9}
			
			& & 15 & 79.4 & 79.4 & 78.1 & 78.6 & \textbf{77.4} & 78.8 \\
			& Pattern 4 & 22 & 81.3 & 78.4 & 78.4 & 80.6 & 79.4 & 80.4 \\
			& & 29 & 79.9 & 79.3 & 79.8 & 79.5 & 79.7 & 81.2 \\
			\hline
			\multirow{12}{*}{0.06 } & & 15 & 81.2 & 80.5 & 79.0 & 77.8 & 78.7 & 79.6 \\
			& Pattern 1 & 22 & 80.0 & 81.7 & 79.8 & 80.7 & 80.5 & 80.2 \\
			& & 29 & 81.2 & 78.7 & 79.2 & 81.2 & 79.7 & 80.1 \\
			\cline{2-9}
			
			& & 15 & 78.7 & 77.5 & 81.4 & 80.7 & 81.0 & 80.7 \\
			& Pattern 2 & 22 & 80.6 & 81.8 & 79.2 & 80.3 & 81.6 & 80.2 \\
			& & 29 & 78.5 & 80.2 & 80.0 & 77.7 & 78.1 & 78.0 \\
			\cline{2-9}
			
			& & 15 & 78.1 & 80.0 & 80.9 & 79.7 & 79.3 & 78.8 \\
			& Pattern 3 & 22 & 81.2 & 80.2 & 80.0 & 78.3 & 82.2 & 81.1 \\
			& & 29 & 79.6 & 81.6 & 79.8 & 80.2 & 81.6 & \textbf{76.9} \\
			\cline{2-9}
			
			& & 15 & 78.2 & 79.8 & 78.9 & 79.5 & \textbf{77.3} & 79.2 \\
			& Pattern 4 & 22 & 79.2 & 81.1 & 79.4 & \textbf{76.8} & 79.2 & 80.4 \\
			& & 29 & 79.9 & 78.5 & 79.8 & 80.1 & 78.9 & 81.8 \\
			\hline
		\end{tabular}
		\begin{tablenotes}
			
			\item    \lq\lq Max\rq\rq is the day in which the maximal proximal effect is attained. $\bar d=(1/T)\sum_{t=1}^{T}Z_t'd$ is the average standardized treatment effect.  Bold numbers are significantly (at .05 level) lower than 0.80.
			
		\end{tablenotes}	
		\label{workfalsecase1full}
	\end{threeparttable}
\end{table}

\begin{table}[H]
	\centering
	\begin{threeparttable}	
		\caption{Simulated Type I error rate($\%$) and power($\%$) when working assumption (c) is violated. The trends of $\bar \sigma_t$ are provided in Figure~\ref{TrendofSigma}. The standardized average effect is 0.1.  $E[I_t] =0.5$. The associated sample sizes are 41 and 42 when the day of maximal effect is 22 and 29.}
		\begin{tabular}{|c|c|cccc|cccc|}

			& &
			&
			\multicolumn{2}{c}{Max = 22} & & &
			\multicolumn{2}{c}{Max = 29} &
			\\
			$\phi$ in AR(1) & $\frac{\sigma_{1t}}{\sigma_{0t}}$& const.& trend 1 & trend 2 & trend 3 & const.& trend 1 & trend 2 & trend 3 \\
			
			\hline
			& 0.8 & 4.1 & 4.3 & 3.3 & 5.4 & 4.7 & 4.9 & 2.8 & 4.1 \\
			-0.6 & 1.0 & 4.6 & 5.0 & 4.0 & 4.4 & 4.4 & 4.8 & 4.2 & 4.3 \\
			& 1.2 & 3.8 & 4.5 & 5.2 & 5.5 & 4.3 & 4.1 & 4.5 & 3.8 \\
			\hline
			& 0.8 & 5.2 & 4.7 & 4.0 & 3.4 & 5.4 & 4.9 & 6.2 & 4.5 \\
			-0.3& 1.0  & 4.9 & 4.5 & 4.5 & 4.3 & 5.2 & 5.1 & 4.0 & 3.7 \\
			& 1.2 & 5.4 & 4.6 & 4.1 & 3.8 & 3.7 & 5.2 & 4.3 & 5.0 \\
			\hline
			& 0.8 & 4.8 & 4.0 & 4.1 & 3.9 & 4.7 & 5.2 & 3.7 & 4.2 \\
			0& 1.0 & 5.4 & 4.0 & 5.8 & 3.9 & 4.1 & 4.0 & 5.9 & 5.7 \\
			& 1.2 & 4.4 & 4.9 & 5.0 & 4.6 & 3.7 & 4.8 & 4.4 & 4.9 \\
			\hline
			& 0.8 & 5.3 & 4.4 & 4.7 & 3.2 & 4.6 & 5.4 & 5.6 & 4.1 \\
			0.3 & 1.0 & 5.5 & 4.0 & 3.4 & 3.7 & 5.0 & 4.6 & 4.0 & 3.6 \\
			& 1.2 & 3.8 & 4.5 & 4.5 & 4.8 & 4.5 & 5.0 & 6.2 & 4.3 \\
			\hline
			& 0.8 & 5.5 & 3.9 & 5.3 & 3.8 & 3.3 & 3.5 & 5.1 & 4.2 \\
			0.6 & 1.0 & 4.0 & 3.7 & 5.2 & 5.1 & 4.8 & 5.1 & 5.0 & 4.7 \\
			& 1.2 & 4.5 & 5.1 & 4.6 & 4.9 & 4.5 & 4.4 & 4.7 & 4.8 \\
			\hline
			\hline

			& 0.8 & 82.8 & 82.7 & 83.7 & 79.9 & 83.6 & 80.6 & 88.7 & 79.2 \\
			-0.6 & 1.0 & 81.1 & 79.1 & 79.9 &\textbf{ 74.8} & 77.7 & \textbf{74.3} & 84.8 & \textbf{70.4} \\
			& 1.2 & \textbf{76.6} & \textbf{76.3} & \textbf{76.3} & \textbf{70.6} & 77.6 & \textbf{72.0} & 80.7 & \textbf{70.4} \\
			\hline
			& 0.8 & 83.0 & 83.0 & 86.0 & 80.3 & 82.7 & 79.2 & 87.9 & 78.0 \\
			-0.3 & 1.0 & 77.6 & 81.4 & 80.7 & \textbf{74.9 }& 79.1 & \textbf{74.5} & 86.0 & \textbf{73.7} \\
			& 1.2 & 78.2 & \textbf{76.9} & \textbf{77.3} & \textbf{73.4} & \textbf{74.4} & \textbf{71.2} & 81.0 & \textbf{70.7 }\\
			\hline
			& 0.8 & 84.6 & 84.6 & 82.1 & 79.0 & 81.8 & 81.5 & 88.0 & 78.0 \\
			0 & 1.0 & 80.1 & 78.6 & 80.9 & \textbf{73.6} & 77.7 & \textbf{76.5} & 86.1 & \textbf{71.8 }\\
			& 1.2 & \textbf{76.0 }& \textbf{76.7} &\textbf{ 77.4} & \textbf{70.6} & \textbf{74.5 }& \textbf{69.9} & 83.4 & \textbf{69.6} \\
			\hline
			& 0.8 & 83.6 & 79.7 & 84.6 & 79.7 & 82.1 & 81.7 & 88.2 & \textbf{75.7} \\
			0.3 & 1.0 & 81.5 & 82.4 & 82.3 & \textbf{73.9} & 79.5 & \textbf{74.6} & 85.1 & \textbf{71.5} \\
			& 1.2 &\textbf{ 74.8} &\textbf{ 76.6} & 78.2 &\textbf{ 71.1 }&\textbf{ 75.5} &\textbf{ 71.1} & 82.5 & \textbf{70.1} \\
			\hline
			& 0.8 & 81.4 & 83.1 & 83.5 & 80.5 & 83.1 & \textbf{77.1 }& 86.6 &\textbf{ 76.9} \\
			0.6 & 1.0  & 80.7 & \textbf{76.4} & 79.0 & \textbf{74.8} & 80.4 & \textbf{73.4} & 84.7 & \textbf{76.8} \\
			& 1.2 &\textbf{ 77.0} & 77.5 & \textbf{77.0} & \textbf{73.5} & \textbf{74.4} & \textbf{72.5} & 81.6 & \textbf{69.4} \\
			\hline			
		\end{tabular}
		\begin{tablenotes}
			\item   $\phi$ is the parameter in AR(1) process for $\{\epsilon_t\}_{t=1}^T$. Bold numbers are significantly(at .05 level) greater than .05 (for type I error)and less than 0.80 (for power).
		\end{tablenotes}
		\label{workfalsecase2T}
	\end{threeparttable}
\end{table}

\begin{table}[H]
	\centering

	\begin{threeparttable}
			\caption{Simulated Type I error rate($\%$) when working assumption (d) is violated. $E[I_t ] = 0.5$. The proximal effect $Z_t'd$ satisfies the average is 0.1 and day of maximal effect is 29. N = 42.}
				
		\begin{tabular}{|c|c|ccc|}

			Parameters in $I_t$ &\diagbox{$\gamma_1$}{$\gamma_2$}&  -0.1 &  -0.2 &  -0.3 \\
			\hline
			&-0.2 & 5.7 & 3.2 & 3.9 \\
			$\eta_1 = -0.1, \eta_2 = -0.1$&-0.5 & 3.2 & 4.2 & 4.9 \\
			&-0.8 & 4.2 & 5.1 & 5.5 \\
			\hline
			&-0.2 & 5.4 & 3.8 & 3.9 \\
			$\eta_1 = -0.2, \eta_2 = -0.1$&-0.5 & 4.4 & 4.4 & 4.8 \\
			&-0.8 & 4.7 & 4.3 & 4.6 \\
			\hline
			&-0.2 & 4.5 & 5.0 & 5.0 \\
			$\eta_1 = -0.1, \eta_2 = -0.2$ &-0.5 & 4.9 & 3.8 & 6.0 \\
			&-0.8 & 4.7 & 4.8 & 4.8 \\
			\hline
		\end{tabular}
		\begin{tablenotes}
		
			\item   $\eta_1, \eta_2$ are parameters in generating $I_t$. $\gamma_1$, $\gamma_2$ are coefficients in the model of  $Y_{t+1}$. All numbers in this table are significantly (at .05 level) greater than .05.
		\end{tablenotes}
		\label{workfalsecase3T}		
	\end{threeparttable}

\end{table}
\vspace{-15mm}
\begin{table}[H]
	\centering
	
	\begin{threeparttable}
		\caption{Degradation in power when average proximal main effect is underestimated. The day of maximal treatment effect is attained at day 29 and the average availability is 0.5 in all cases. The associated sample sizes for each value of average treatment effect are provided in first column.}
		\begin{tabular}{|c|c|cccc|}

			\multirow{2}{6.5em}{\ $\bar d$ in Sample \ Size Formula}&\multirow{2}{*}{True $\bar d$}   &\multicolumn{4}{c|}{Availability Pattern } \\
			\cline{3-6}
			&  & Pattern 1 & Pattern 2 & Pattern 3 & Pattern 4\\
			
			\hline	
			\multirow{11}{*}{0.10 (N = 42)}& 0.098 & 76.2 & 78.9 & 77.6 & 78.6 \\
			& 0.096 & 75.1 & 74.6 & 78.8 & 74.0 \\
			& 0.094 & 73.7 & 70.7 & 75.4 & 73.4 \\
			& 0.092 & 71.5 & 71.6 & 73.2 & 71.6 \\
			& 0.090 & 68.9 & 68.4 & 69.6 & 67.3 \\
			& 0.088 & 65.4 & 65.6 & 66.1 & 65.7 \\
			& 0.086 & 66.4 & 67.9 & 65.2 & 66.7 \\
			& 0.084 & 62.3 & 63.4 & 63.0 & 59.6 \\
			& 0.082 & 60.0 & 60.2 & 60.5 & 58.2 \\
			& 0.080 & 58.9 & 59.8 & 57.8 & 61.4 \\
			\hline
			\multirow{11}{*}{0.08(N = 64)}& 0.078 & 78.2 & 80.2 & 76.8 & 75.8 \\
			& 0.076 & 77.3 & 76.7 & 76.2 & 75.4 \\
			& 0.074 & 73.1 & 72.2 & 71.2 & 71.4 \\
			& 0.072 & 70.7 & 71.0 & 69.4 & 68.2 \\
			& 0.070  & 68.2 & 66.0 & 65.2 & 66.1 \\
			& 0.068 & 65.5 & 64.3 & 64.6 & 65.7 \\
			& 0.066 & 62.8 & 62.3 & 61.8 & 59.4 \\
			& 0.064 & 61.9 & 58.5 & 59.5 & 62.1 \\
			& 0.062 & 53.9 & 52.6 & 57.0 & 56.9 \\
			& 0.060 & 54.6 & 51.1 & 54.8 & 53.4 \\
			\hline
			\multirow{11}{*}{0.06(N = 109)}& 0.058 & 75.6 & 76.9 & 74.0 & 78.1 \\
			& 0.056  & 73.9 & 73.1 & 73.1 & 72.7 \\
			& 0.054 & 68.6 & 71.1 & 69.3 & 68.5 \\
			& 0.052 & 65.4 & 69.4 & 63.6 & 66.8 \\
			& 0.050 & 61.0 & 62.8 & 64.1 & 63.2 \\
			& 0.048 & 57.4 & 58.6 & 56.4 & 56.1 \\
			& 0.046 & 53.6 & 53.4 & 52.9 & 54.8 \\
			& 0.044 & 52.0 & 48.9 & 50.1 & 53.0 \\
			& 0.042	& 45.7 & 43.9 & 44.9 & 46.4 \\
			& 0.040 & 40.4 & 42.2 & 42.3 & 42.7 \\ 	
			\hline
		\end{tabular}
		\label{wrongguessave}
	\end{threeparttable}
\end{table}

\begin{table}[H]
	\centering
	\begin{threeparttable}
		\caption{Degradation in power when average availability is underestimated. The day of maximal treatment effect is attained at day 29 and the average proximal main effect is 0.1 in all cases. The associated sample sizes are given in first column.}
		\begin{tabular}{|c|c|cccc|}

			$(1/T)\sum_{t=1}^{T}\tau_t$ in  & True &\multicolumn{4}{c|}{Availability Pattern } \\
			\cline{3-6}
			Sample Size Formula & $(1/T)\sum_{t=1}^{T}\tau_t$ & Pattern 1 & Pattern 2 & Pattern 3 & Pattern 4 \\
			\hline

			\multirow{11}{*}{0.5 (N = 42)}& 0.048  & 76.4 & 81.7 & 76.0 & 78.2 \\
			& 0.046 & 73.9 & 75.5 & 73.6 & 75.8 \\
			& 0.044 & 70.6 & 72.1 & 71.0 & 71.7 \\
			& 0.042 & 70.8 & 70.6 & 74.2 & 70.3 \\
			& 0.040  & 70.3 & 69.2 & 65.7 & 68.6 \\
			& 0.038  & 66.0 & 66.8 & 67.8 & 67.0 \\
			& 0.036  & 64.0 & 62.5 & 62.4 & 62.9 \\
			& 0.034 & 60.8 & 61.3 & 59.4 & 63.9 \\
			& 0.032& 56.4 & 59.2 & 54.7 & 59.8 \\
			& 0.030 & 51.4 & 53.1 & 51.9 & 54.5 \\
			\hline

			\multirow{11}{*}{0.7 (N = 32)}& 0.068 & 79.5 & 76.1 & 79.1 & 75.0 \\
			& 0.066 & 77.3 & 75.7 & 74.0 & 76.4 \\
			& 0.064 & 74.5 & 74.7 & 73.5 & 77.1 \\
			& 0.062 & 73.2 & 73.0 & 75.1 & 72.5 \\
			& 0.060 & 69.8 & 70.5 & 73.5 & 72.5 \\
			& 0.058 & 71.0 & 69.6 & 71.3 & 67.3 \\
			& 0.056  & 68.8 & 70.3 & 66.6 & 64.0 \\
			& 0.054  & 68.1 & 65.8 & 65.3 & 68.6 \\
			& 0.052 & 62.4 & 64.9 & 65.6 & 62.9 \\
			& 0.050  & 60.6 & 63.3 & 62.8 & 61.4 \\ 		
			\hline
		\end{tabular}
		\label{wrongguessavetau}
	\end{threeparttable}
\end{table}
\end{appendices}

\end{document}